\newif\ifINTERNAL \INTERNALtrue \INTERNALfalse
\newif\ifjbf \jbftrue \jbffalse
\newif\ifANONYMOUS \ANONYMOUStrue \ANONYMOUSfalse

\ifjbf
\documentclass[preprint,12pt,authoryear]{elsarticle} 
\else \documentclass[11pt,a4paper]{article} \fi

\usepackage{whale} 
\graphicspath{{_figures/}}

\definecolor{BrickRed}{rgb}{.625,.25,.25}

\definecolor{markergreen}{rgb}{0.6, 1.0, 0}
\providecommand{\marker}[1]{\fcolorbox{markergreen}{markergreen}{{#1}}}
\usepackage[normalem]{ulem}

\usepackage{mathrsfs}

\newcommand{\E}{{\mathbb{E}}}

\newcommand{\p}{{\bf P}}

\newcommand{\e}{{\bf e}}

\newcommand{\dd}{{\rm d}}

\providecommand{\Ncdf}{{\rm N}}

\newcommand{\1}{\ensuremath{\mathbf{1}}}

\ifx\prop\undefined
\newtheorem{prop}{Proposition}[section]
\fi

\newcommand{\argmax}{\operatornamewithlimits{argmax}}

\providecommand{\abstractText}{
  \noindent In 2012, JPMorgan accumulated a USD~6.2 billion loss on a
  credit derivatives portfolio, the so-called ``London Whale'', partly
  as a consequence of de-correlations of non-perfectly correlated
  positions that were supposed to hedge each other. Motivated by this
  case, we devise a factor model for correlations that allows for
  scenario-based stress testing of correlations. We derive a number of
  analytical results related to a portfolio of homogeneous
  assets. Using the concept of Mahalanobis distance, we show how to
  identify adverse scenarios of correlation risk.  In addition, we
  demonstrate how correlation and volatility stress tests can be
  combined. As an example, we apply the factor-model approach to the
  ``London Whale'' portfolio and determine the value-at-risk impact
  from correlation changes. Since our findings are particularly
  relevant for large portfolios, where even small correlation changes
  can have a large impact, a further application would be to stress
  test portfolios of central counterparties, which are of systemically
  relevant size.}

\ifjbf 
\journal{Journal of Banking \& Finance}
\begin{document}
\begin{frontmatter}
  \title{\Large\bf A factor-model approach for correlation scenarios
    and correlation stress testing%
    \ifANONYMOUS\else \tnoteref{thanks} \fi }
  
  \ifANONYMOUS
  \author{}
  \else \tnotetext[thanks]{ Part of this
    research was undertaken while Fabian Woebbeking was visiting
    Columbia University, NYC.  This work is financially supported by
    the Global Association of Risk Professionals (GARP), by the
    Europlace Institute of Finance (EIF) -- Labex Louis Bachelier and
    by The Frankfurt Institute of Risk Management and Regulation
    (FIRM). The authors would like to thank Stephen Taylor (NJIT) and
    two anonymous referees for helpful comments and
    discussion. Declarations of interest: none.  }
  
  \author[fsfm]{N.~Packham}
  \ead{packham@hwr-berlin.de}
  \author[hof]{C.F.~Woebbeking\corref{cor1}}
  \ead{woebbeking@finance.uni-frankfurt.de}
  \address[fsfm]{Berlin School of Economics and Law, Badensche Str.\
    52, 10825 Berlin, Germany}
  \address[hof]{Goethe University Frankfurt, Grueneburgplatz 1,
    60323\\Frankfurt, Germany} 
  \cortext[cor1]{Corresponding author: Tel.: +49 (69) 798 33731; Fax:
    +49 (69) 798 33901}
  \fi

\begin{abstract}
  \abstractText
\end{abstract}

\end{frontmatter}
\else 

\title{ \Large\bf A factor-model approach for correlation scenarios
  and correlation stress testing\thanks{ Part of this research was
    undertaken while Fabian Woebbeking was visiting Columbia
    University, NYC.  This work is financially supported by the Global
    Association of Risk Professionals (GARP), by the Europlace
    Institute of Finance (EIF) -- Labex Louis Bachelier and by The
    Frankfurt Institute of Risk Management and Regulation (FIRM). The
    authors would like to thank Stephen Taylor (NJIT) for helpful
    comments and discussion as well as two anonymous referees for
    their valuable suggestions that helped to improve the paper in 
      several respects.
  }
}

\author{ Natalie Packham\thanks{Berlin School of Economics and Law,
    \href{mailto:packham@hwr-berlin.de}{packham@hwr-berlin.de}} \and
  Fabian Woebbeking\thanks{Goethe University Frankfurt,
    \href{mailto:woebbeking@finance.uni-frankfurt.de}
    {woebbeking@finance.uni-frankfurt.de}} }

\begin{document}
\ifINTERNAL
\clearpage
\begin{center}
\marker{INTERNAL PAGE}
\end{center}
\tableofcontents
\clearpage
\fi

\maketitle

\begin{abstract}

\abstractText \bigskip

\smallskip

\end{abstract}
\fi 

\section{Introduction}

Diversification -- typically captured by correlation -- lies at the
heart of many financial applications: a diversified portfolio is less
risky than a concentrated portfolio; hedging strategies may involve
only imperfectly correlated assets instead of perfect substitutes. It
is well-known that correlations are not constant over time and may be
strongly affected by specific events
\citep{Karolyi1996,longin2001extreme,Ang2002a, wied2012testing,
  pu2012correlation, adams2017correlations}. Changes in correlation
may lead to potentially unexpected or unquantified losses, see e.g.\
LTCM \citep{Jorion2000}, Amaranth Advisors \citep{Chincarini2007}.

This paper develops a technique for generating correlation matrices
from specific risk factor scenarios. The method allows to challenge
diversification benefits in a realistic way by quantifying potential
losses from correlation changes or a correlation break-down due to
various scenarios. Consequently, worst-case scenarios and their impact
can be identified. Quantifying these risks is particularly important
if a portfolio or a hedging strategy may be adversely affected by a
correlation breakdown amongst the portfolio constituents. For example,
hedging strategies involving non-perfect substitutes, such as a stock
portfolio and index futures for hedging, are sensitive to correlation
changes and thus vulnerable to adverse correlation scenarios.

The technique borrows elements from parameterising correlation
matrices in interest rate modelling, e.g.\
\cite{Rebonato2002,Brigo2002,Schoenmakers2003}. These
parameterisations have in common that the degree of correlation
depends on the difference in maturity of the underlying interest rates
(e.g.\ swap rates). In its simplest form, correlations are determined
by $\e^{-\beta|i-j|}$, where $\beta>0$ is a constant parameter, and
$i,j$ are maturities. This captures the stylised fact that
correlations decay with increasing maturity difference.

In this paper, this approach is generalised by defining factors that
characterise differences in the assets under consideration, and by
parameterising correlations via ``distances'' capturing these
differences. The parameters themselves can be calibrated for example
from historical data. Scenarios are generated by varying the
parameters, where an increase in a parameter captures a de-correlation
related to a factor.

The method is capable of identifying the factor structure of worst
case scenarios. More specifically, given the mapping of correlation
risk factor to risk measure, one can find the global maximum of the
risk measure and infer the corresponding risk factor scenario.  As
each parameter represents an economically relevant correlation risk
factor, it is therefore possible to identify critical portfolio
structures that might require particular attention from a risk
management perspective.

Aside from the impact of a given scenario, one is also interested in
the plausibility of the chosen scenarios. This can be implemented by
assigning a joint probability distribution to the correlation
parameters in order to define a constraint for correlation
scenarios. In this paper, the constraint is specified via the
so-called Mahalanobis distance, which measures the distance of a
normally distributed random vector from the center of the
distribution.

As correlation stress often occurs jointly with volatility shocks, we
also demonstrate how to combine the two stress scenarios. To model
volatility separately we assume that asset returns follow a
multivariate Student $t$-distribution (as opposed to a normal
distribution). As a $t$-distribution can be conveniently decomposed
into a correlated normal distribution component and an
inverse-gamma-distributed scaling factor, volatility stress is
introduced by setting the scaling factor to a given quantile.

To demonstrate the technique, correlation stress tests are applied to
the portfolio of the so-called ``London Whale'', a term used in the
finance industry to denote a USD 6.2~billion loss in 2012 of a credit
derivative portfolio run by JPMorgan. In late 2011, in an effort to
reduce the risk of the position without monetising losses, the
notional amount of the portfolio was increased, while relying on the
ability of similar credit index positions to act as hedges for each
other \citep{JPMorgan2013, USS2013Report}. Our analysis shows that
correlation scenarios and stress tests reveal the high riskiness of
this portfolio and thus might have led to a more appropriate risk
assessment of the portfolio.\footnote{%
  Other risks specific to the ``London Whale'', especially concerning
  the size of the position relative to market size, are treated in
  \citet{Cont2015}.}%

A further application where correlation scenario and stress testing
can reveal inherent risks is the practice of so-called ``portfolio
margining'' in initial margin calculations of clearing houses. Here,
netting of offsetting positions reduces the margin
requirement. However, when positions are not perfect hedges, but only
highly correlated, an adverse correlation scenario could lead to
substantial margin calls, thereby increasing counterparty risk at a
systematic level.

The literature on establishing correlation stress tests is scarce,
even though it is well established that correlations are not constant
over time and may be strongly affected by specific events
\citep{longin2001extreme, wied2012testing,
  pu2012correlation}. \citet{adams2017correlations} observe that
correlations vary over time and, in addition, experience level shifts
and structural breaks that occur in response to economic or financial
shocks. \citet{krishnan2009correlation} and
\citet{mueller2017international} provide empirical evidence that
investors demand a correlation risk premium, which is related to the
uncertainty about future correlation
changes. \citet{buraschi2010correlation} develop a framework for
inter-temporal portfolio choice that includes hedging components
against correlation risk. 

The prominent role of correlation in financial portfolios led to
regulatory agencies calling for risk model stress tests that account
for ''significant shifts in correlations'' \citep[p.~207
ff.]{bcbs128}. However, there is little literature on parametric
correlation modelling, in particular related to risk-factor driven
stress testing: aside from challenges of mathematical consistency in
correlation modelling (correlation matrices must be positive
semi-definite, see e.g.\ \citet{qi2010correlation,ng2014black} for
solutions to this problem), the specification of stressful yet
plausible scenarios for correlations is far from straightforward.

The selection of plausible scenarios poses a challenge in the
development of stress testing methods in general. The use of
historical or hypothetical scenarios is problematic, as the
probability and thus the plausibility of a
scenario is typically 
unknown, while at the same time relevant scenarios might be
neglected. In an extensive study, \citet{alexander2008developing}
compare various well-known models in their ability to conduct
meaningful stress tests.  \citet{glasserman2015stress} develop an
empirical likelihood approach for the selection of stress scenarios,
with a focus on reverse stress testing. \citet{kopeliovich2015robust}
present a reverse stress testing method to determine scenarios that
lead to a specified loss level.
\citet{Breuer2009} and \citet{flood2015systematic} use the Mahalanobis
distances to select scenarios from a multivariate distribution of risk
factors. 

\citet{breuer2013systematic} extend these approaches and consider
various application scenarios, amongst them stressed default
correlations, which refer to the correlations of Bernoulli variables
denoting the default or survival of loans or
obligors. \citet{Studer1999} considers correlation breakdowns by
identifying the worst-case correlation scenario in a constrained
region of P\&L scenarios. However, solving the problem turns out to be
intractable in the sense that it is NP-hard. Also, the likelihood or
plausibility of such a correlation scenario is not known. The
difference in our setting is that we model correlation itself in a
parametric way and -- imposing a distribution assumption on the risk
factors driving correlation, e.g.\ calibrated from historical data --
find the risk-factor scenario that produces the worst loss within a
given range of plausible correlation scenarios.

The paper is structured as follows: In Section~\ref{sec:methodology}
we present the correlation stress testing methodology alongside
analytical results for both, the stress test and scenario selection
procedures. Section~\ref{sec:applicationtowhale} consists of a concise
review of the ``London Whale'' case as well as the results from
correlation stress testing the credit portfolio using the methods
developed in the previous section. Section~\ref{sec:conclusion}
concludes. A detailed review of the London Whale case is provided as
an online appendix at \ifANONYMOUS [anonymous version: Supplement -
Whale Story.pdf].  \else
\href{https://ssrn.com/abstract=3210536}{https://ssrn.com/abstract=3210536}.
\fi

\section{Correlation parameterisation and stress
  testing}
\label{sec:methodology} 

\subsection{Factor model}
The principal idea behind the correlation parameterisation developed
in this paper is to split portfolio correlations into dependence
contributions associated with several risk factors. With each risk
factor a parameter determining the degree of de-correlation on the
overall correlation is associated. Calibrating these parameters and
then adjusting them allows to translate specific economic scenarios
into changes on correlations.

More precisely, let $C$ be an $n\times n$-correlation matrix (i.e.,
positive semi-definite, symmetric, with entries in $[-1,1]$, and ones
on the diagonal) related to the returns of $n$ financial
instruments. In the context of the London Whale position analysed
later, the entries of $C$ are the correlations of credit index spread
returns and related tranche spread returns.\footnote{%
  In this setting correlations are not implied tranche correlations,
  which are used for pricing, but historical spread return
  correlations, as would be used in risk management.} In the London
Whale case these are typically positive (with only few exceptions near
zero, which are set to a small positive constant), and generally we
assume that all correlations are in $(0,1]$.

The factors that determine the correlations are denoted by
$\mathbf{x}=(x^1, \ldots, x^m)'$. In the context of the London Whale
position, the factors include the maturity, the index series, a dummy
variable determining whether the security is investment grade or not,
and others. Further choices could be factors relating to geographical
regions, industries or balance sheet data. Correlations $c_{ij}$ of
securities $i$ and $j$ are modelled as
\begin{equation*}
  c_{ij} = \exp\left(-(\beta_1 |x^1_i-x^1_j| + \beta_2 |x_i^2 - x^2_j| +
    \cdots + \beta_m |x^m_i-x^m_j|)\right), \quad i,j=1,\ldots, n,
\end{equation*}
with $\beta_1,\ldots, \beta_m$ positive coefficients, the parameters.
This is the simplest, most parsimonious functional form relating
differences in the risk factors with the correlations of the
securities. It implies that the greater the distance $|x_i^k-x_j^k|$,
the greater the de-correlation amongst the securities $i$ and $j$. If
two instruments are identical in all respects, then they are assigned
a correlation of $1$. With additional information about the
relationship between risk factors and correlations, other, more
complex functional forms may be feasible.  Similar approaches to
parameterising correlation matrices are common in interest-rate
modelling, see e.g.\ \citet{Schoenmakers2003,Rebonato2004,Brigo2006}.

In the simple model above, given historical returns, the parameters
$\beta_1,\ldots, \beta_m$ are easily determined by standard regression
techniques such as OLS on the transformed correlations
$-\ln(c_{ij})$. A scenario such as ``the correlation between
investment grade and high-yield securities decreases'' is then
implemented by increasing the corresponding $\beta$-parameter. With
parameters calibrated on a regular basis, the parameter history can be
used to obtain reasonable scenarios.

\subsection{Stress testing a homogeneous portfolio}
\label{sec:stress-test-homog}

To better understand the stress testing effect we consider a stylised,
homogeneous portfolio and derive closed formulas for the impact of
various correlation stress test scenarios. A homogeneous portfolio
reduces the number of parameters involved to a minimum and therefore
allows for a general understanding of the behaviour under stress.
This is similar to analysing diversification effects of a homogeneous
portfolio in standard Markowitz portfolio theory.

The setup is as follows: the $m$ risk factors are binary in the sense
that they express properties that are either present or absent in a
security. The number of securities is $n=2^m$ and they exhibit all
$2^m$ combinations of risk factor combinations (for example, one could
set
$(\1_{\{x_i^1\not=x_j^1\}}, \ldots, \1_{\{x_i^m\not=x_j^m\}})=
(i-1)\oplus (j-1)$, with $(i-1)\oplus (j-1)$ the bitwise XOR operator
of the binary representations of $i-1$ and $j-1$). As a consequence,
no two securities are equal in terms of their risk factor
exposure. The securities all have equal volatility, and the portfolio
is equally-weighted. However, as a consequence of choosing binary risk
factor combinations, the correlations are not homogeneous.

We assume that risk of the portfolio is measured by value-at-risk
(VaR) in a {\em variance-covariance approach}, i.e.,
\begin{equation}
  \label{eq:5}
  \text{VaR}_\alpha = -\Ncdf_{1-\alpha} \cdot  V_{0} \cdot
  \left(\mathbf w^\intercal\, \boldsymbol\Sigma\, \mathbf w\right)^{1/2},
\end{equation}
where $\Ncdf_{1-\alpha}$ denotes the $(1-\alpha)$-quantile of the
standard normal distribution, $V_{0}$ denotes the current position
value, $\mathbf w$ is the vector of portfolio weights and
$\boldsymbol\Sigma$ denotes the covariance matrix of the portfolio
components' returns with entries $\sigma^2$ along the diagonal
denoting the stand-alone variances of the assets. In this setting we
assume that the expected return is zero, which is a reasonable
assumption for short time horizons.

The normal distribution assumption can easily be generalised, e.g.\ to
a Student $t$-distri\-bution. We will use this more generalised setup in
Section \ref{sec:stress-test-corr} when we combine correlation and
volatility stress testing.

\begin{proposition}
  The portfolio variance is given by
  \begin{equation}
    \label{eq:9}
    \mathbf w^\intercal\, \boldsymbol\Sigma\, \mathbf w 
    = \frac{\sigma^2}{n} \prod_{k=1}^m \left(1+\e^{-\beta_k}\right),
  \end{equation}
  The average correlation amongst pairwise different asset returns is
  \begin{equation}
    \label{eq:10}
    \overline \rho(\boldsymbol\beta) = \frac{1}{(n-1)} \prod_{k=1}^m
    \left(1+\e^{-\beta_k}\right) -\frac{1}{n-1}. 
  \end{equation}
\end{proposition}

\begin{proof} 
  Assuming without loss of generality that the bitwise XOR operator
  described above defines the value of the indicator variables, the
  portfolio variance simplifies to
  \begin{align*}
    \mathbf w^\intercal\, \boldsymbol\Sigma\, \mathbf w 
    &= \frac{1}{n^2} \sum_{i=1}^n \sum_{j=1}^n \e^{-\sum_{k=1}^m
      \beta_k \cdot \1_{\{x_i^k\not=x_j^k\}}}\sigma^2\\ %
    &= \frac{\sigma^2}{n^2} \sum_{i=1}^n \left[\sum_{j=1}^{n/2}
      \e^{-\sum_{k=1}^{m-1} \beta_k  \cdot
      \1_{\{x_i^k\not=x_j^k\}} - \beta_m\cdot 0} + \sum_{j=n/2+1}^n
      \e^{-\sum_{k=1}^{m-1} \beta_k  \cdot \1_{\{x_i^k\not=x_j^k\}} -
      \beta_m\cdot 1}\right]\\%
    &= \frac{\sigma^2}{n^2} \sum_{i=1}^n \left(1+\e^{-\beta_m}\right)
      \cdot \sum_{j=1}^{n/2}  \e^{-\sum_{k=1}^{m-1} \beta_k  \cdot
      \1_{\{x_i^k\not=x_j^k\}}}.
  \end{align*}
  Iterating this calculation $m=\log_2 n$ times gives
  \begin{equation*}
    \mathbf w^\intercal\, \boldsymbol\Sigma\, \mathbf w %
    = \frac{\sigma^2}{n^2} \sum_{i=1}^n \prod_{k=1}^m
    \left(1+\e^{-\beta_k}\right) %
    = \frac{\sigma^2}{n} \prod_{k=1}^m
    \left(1+\e^{-\beta_k}\right). %
  \end{equation*}
  The average correlation is given as
  \begin{equation*}
    \frac{1}{n(n-1)} \sum_{i=1}^n \sum_{j=1,j\not=i}^n \e^{-\sum_{k=1}^m
      \beta_k \cdot \1_{\{x_i^k\not=x_j^k\}}} %
    =   \frac{1}{n(n-1)} \left[\sum_{i=1}^n \sum_{j=1,}^n \e^{-\sum_{k=1}^m
        \beta_k \cdot \1_{\{x_i^k\not=x_j^k\}}} - \sum_{i=1}^n
      1\right],
  \end{equation*}
  and the claim follows from the first part of the Proposition.
\end{proof}

\begin{corollary}
  The sensitivity of the variance with respect to a single
  $\beta$-factor $\beta_l$ is
  \begin{equation*}
    \frac{\partial \mathbf w^\intercal \boldsymbol \Sigma \mathbf
      w}{\partial \beta_l} %
    = -\frac{\sigma^2}{n^2} \, \e^{-\beta_l}\cdot \prod_{k\not=l}
    \left(1+\e^{-\beta_k}\right).
  \end{equation*}
  If the $\beta$-factors are homogeneous, i.e.,
  $\beta_1=\cdots=\beta_m=\beta$, then the overall sensitivity is
  \begin{equation*}
    \frac{\partial \mathbf w^\intercal \boldsymbol \Sigma \mathbf w}
    {\partial \beta} %
    = \frac{\partial \frac{\sigma^2}{n^2}
      \left(1+\e^{-\beta}\right)^m} {\partial \beta} %
    = - \frac{\sigma^2}{n} \frac{m\,
      \left(1+\e^{-\beta}\right)^m} {1+\e^\beta}. 
  \end{equation*}
\end{corollary}

For stress testing we associate a probability distribution with
$\boldsymbol\beta$. This allows to formulate scenarios in
probabilistic terms and determine their impact. We assume that the
risk factors themselves are homogeneous, i.e., $\beta_k=\beta$, for
all $k=1,\ldots, m$, and that changes in the risk factor coefficients
$\Delta \boldsymbol\beta$ are jointly normally distributed, each one
with mean $0$ and variance $\sigma_\beta^2$ and with correlations
$\rho_\beta$.

Following e.g.\ \cite{Kupiec1998}, we define a stress scenario on one
set of (``core'') risk factors and, assuming that the given covariance
matrix is unaltered by the stress scenario, set the remaining
(``peripheral'') risk factors to their optimal estimates conditional
on the scenario. Let $\boldsymbol\beta_s$ denote the $j<m$ core factor
parameters that are stressed directly. The remaining $m-j$ peripheral
risk factor parameters $\boldsymbol\beta_u$ are affected by the stress
scenario only indirectly. Under the normal distribution setting above,
it holds that the optimal estimator of $\Delta \boldsymbol\beta_u$
conditional on $\Delta \boldsymbol\beta_s$ is (e.g.\ Theorem \S 13.2
of \citealp{Shiryaev1996}):
\begin{equation*}
  \E(\Delta \boldsymbol\beta_u|\Delta \boldsymbol\beta_s) =
  \Sigma_{us} \Sigma_{ss}^{-1} \Delta \boldsymbol\beta_s,
\end{equation*}
where $\Sigma_{us}$ and $\Sigma_{ss}$ denote the covariance and
variance matrices of $\boldsymbol\beta_u$ and $\boldsymbol\beta_s$.

\begin{proposition}
  \label{prop:correlationstress}
  The portfolio variance when $j$ of the $\beta$-risk factors
  coefficients are stressed by $\Delta\beta$ is given by
  \begin{equation*}
    \mathbf w^\intercal\, \boldsymbol\Sigma\, \mathbf w 
    = \frac{\sigma^2}{n} \left(1+\e^{-(\beta+\Delta \beta)}\right)^j
    \cdot \left(1+\e^{-\left(\beta + \frac{j\cdot\rho_\beta}
          {(j-1)\rho_\beta+1} \Delta \beta\right)}\right)^{m-j}. 
  \end{equation*}
\end{proposition}
\begin{proof}
  It is easily verified that the entries of the $(m-j)\times j$-matrix
  $\Sigma_{us}\Sigma_{ss}^{-1}$ are
  $\displaystyle\frac{\rho_\beta}{(j-1)\rho_\beta+1}$.
\end{proof}

\begin{figure}[t] 
  \centering \includegraphics[scale=.74]{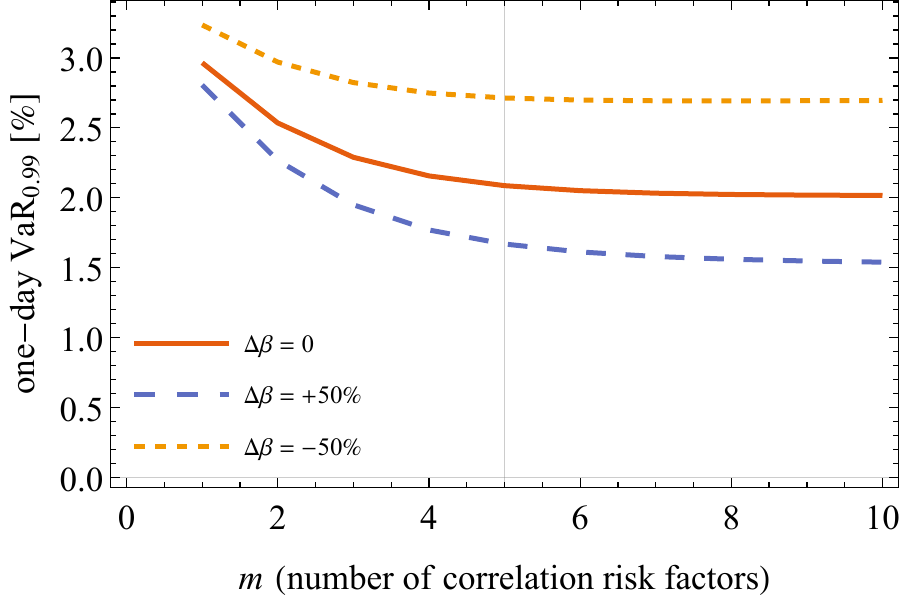}\ \ \ \
  \ \
  \includegraphics[scale=.75]{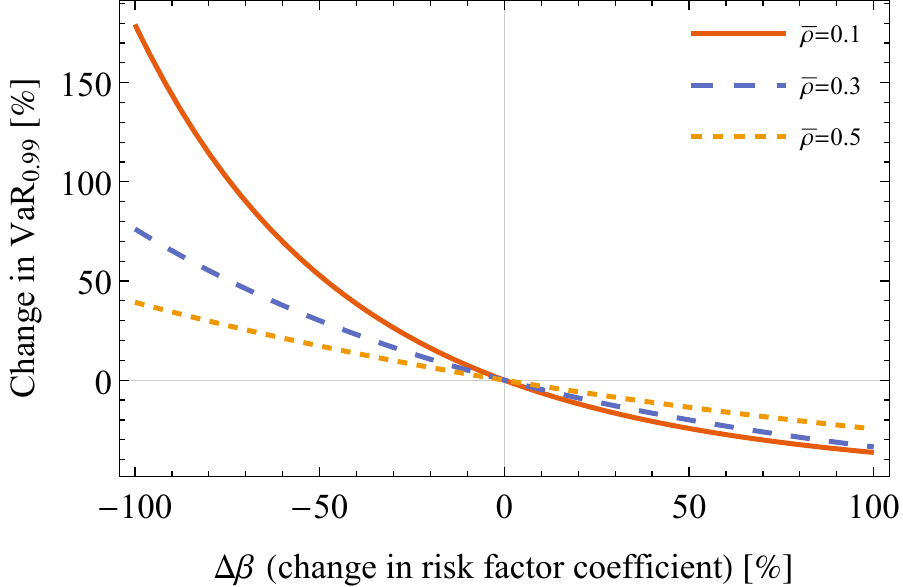}\\
  \vspace{.25cm} \includegraphics[scale=.75]{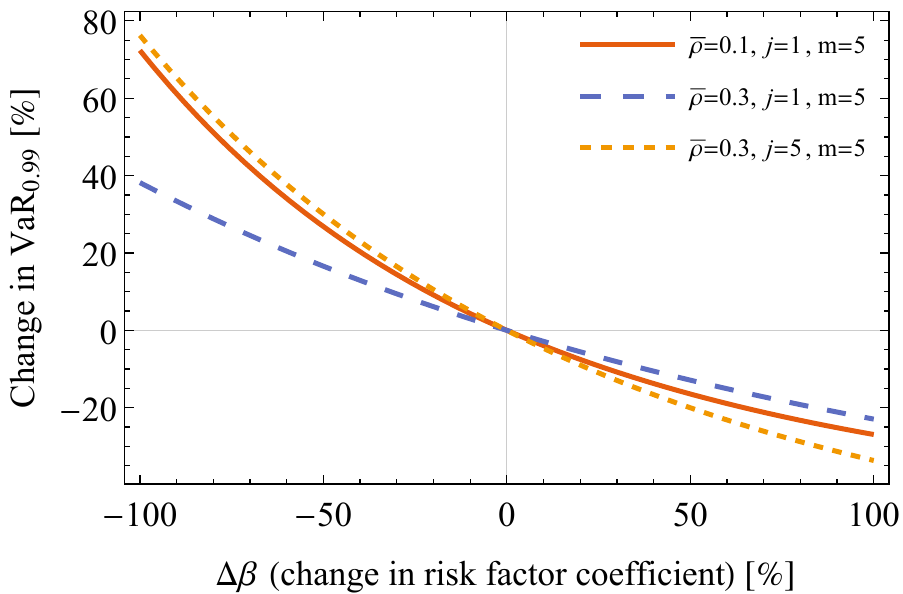}\
  \ \ \ \ \ \includegraphics[scale=.75]{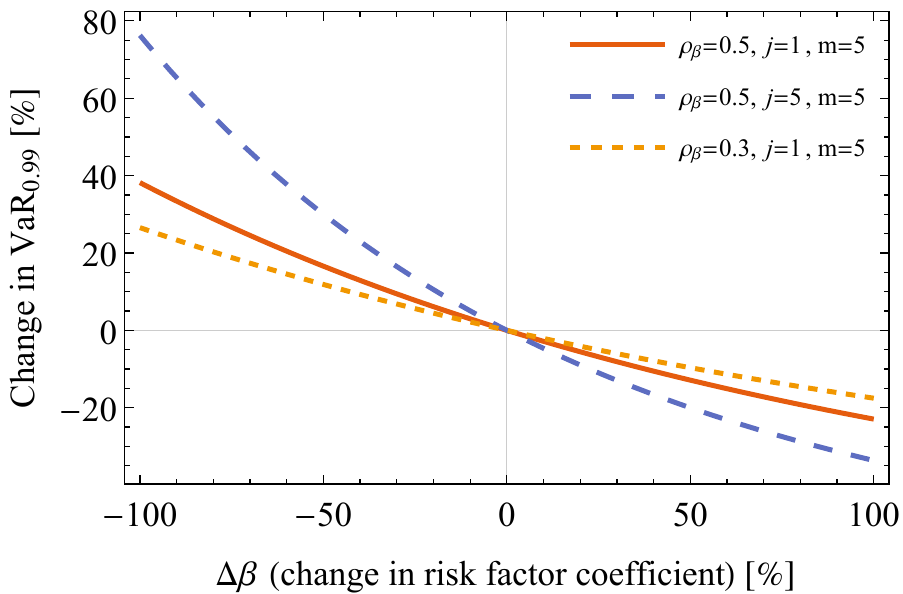}
  \caption{ Top left: Portfolio VaR as a function of the number of
    risk factors $m$ (cf. Proposition 1).  Top right: Change in
    portfolio VaR as a function of $\Delta\beta$ (cf. Proposition 1),
    calibrated to different initial average asset correlations
    $\overline\rho \in \{0.1,0.3,0.5\}$.  Bottom left: Change in
    portfolio VaR as a function of $\Delta\beta$ with peripheral risk
    factor changes that correlate with $\rho_\beta = 0.5$
    (cf. Proposition 2).  Bottom right: Change in portfolio VaR as a
    function of $\Delta\beta$ with correlated peripheral risk factor
    changes (cf. Proposition 3).  All graphs show a 99\%
    Value-at-Risk, and the initial and unstressed $\beta$ is
    calibrated to an average asset correlation of
    $\overline\rho\approx0.3$, unless indicated otherwise (cf.\
    Equation (\ref{eq:10})).}
  \label{fig:propositionsrecalibrated}
\end{figure}

To illustrate the impact of a correlation stress, we apply the above
results to a portfolio of $n=2^m$ assets, where each asset has an
annualised volatility of $\sigma=0.25$ and the average asset
correlation in the portfolio is $0.3$ unless specified
differently. The risk factor coefficient $\beta$ is calibrated to
reflect the target asset correlation for the number of factors $m$,
e.g. with $m=5$ this is achieved by $\beta=0.5204$ (cf.\ Equation
(\ref{eq:10})). Figure \ref{fig:propositionsrecalibrated} shows that
the portfolio VaR decreases as a function of both the number of
correlation risk factors $m$ and the risk factor coefficient
$\beta$. Both results are not surprising as the increasing number of
assets diversifies away idiosyncratic risk leaving only systematic
risk, which in addition, is lowered by an increasing $\beta$. This
result is driven by the long-only structure of the portfolio. The risk
of a portfolio that contains hedging positions may behave differently,
as will be observed in the London Whale example in
Section~\ref{sec:applicationtowhale}.
  
The impact of changes in $\beta$ to the overall portfolio VaR is shown
in the top right graph of
Figure~\ref{fig:propositionsrecalibrated}. For this long-only
portfolio the VaR decays with decreasing correlation. The lower
boundary of $0$ for $\beta$, i.e. a reduction of 100\% in
Figure~\ref{fig:propositionsrecalibrated}, identifies the worst case
scenario.  The impact of $\Delta\beta$ is asymmetric, where the risk
of correlation changes is higher than the benefit. Diversification
benefits through changes in correlation are possible for diversified
portfolios under non-perfect correlation. For hedging portfolios,
which often rely on very high or almost perfect correlation,
correlation changes are most often undesirable and considered a risk.
  
The steepness or sensitivity of VaR to changes in $\beta$ is
determined by the initial average asset correlation, the correlation
between risk factors, as well as the number of factors $j \leq m$ that
are initially stressed (see bottom graphs in
Figure~\ref{fig:propositionsrecalibrated}). In other words, an
initially high average asset correlation reduces the risk of changing
correlations, because the portfolio is already poorly diversified. The
VaR impact when stressing a subset of $j$ factors increases with
$\rho_\beta$, which captures the dependence between correlation risk
factors.

\subsection{Joint stress test of correlation and volatility}
\label{sec:stress-test-corr}

It is well documented that large changes in correlation coincide with
volatility shocks, e.g. \citep{Alexander2008,
  longin2001extreme, Loretan2000}. To this end, we develop a simple
technique that combines both stress scenarios. The principal idea is
to assume that a $d$-dimensional vector of asset returns $\mathbf X$
follows a Student $t$-distribution,
$\mathbf X\sim t(\mathbf {\tilde\Sigma}, \nu)$, with $\nu>2$
and with $\tilde \Sigma$ a matrix describing the dependence as
  explained further below, where we
assume for simplicity that expected asset returns are zero. Then,
$\mathbf X$ follows a 
normal variance mixture distribution with decomposition (cf.\ Chapter
6.2 of \citet{McNeil2015}) 
\begin{equation*}
  \mathbf X=\sqrt{V}\cdot A\cdot\mathbf Z,
\end{equation*}
where $\mathbf Z\sim \Ncdf(0,I_k)$, i.e., $\mathbf Z$ is a vector of
independent standard normally distributed random variables, $V$ is
independent of $\mathbf Z$ and $V\sim \text{Ig}(1/2\, \nu,1/2\,\nu)$,
i.e., the mixing variable $V$ follows an inverse gamma distribution,
and $A$ is a $d\times k$ matrix such that
$\mathbf {\tilde\Sigma}=AA^T$. Because of
$\E V=\displaystyle \frac{\nu}{\nu-2}$, the covariance matrix of
$\mathbf X$ is
$\mathbf \Sigma = \displaystyle\frac{\nu}{\nu-2}
\mathbf{\tilde\Sigma}$ (note that the expectation and covariance
matrix are defined only if $\nu>2$). The correlation matrices of
$\mathbf X$ and $A\mathbf Z$ are the same.

Under the assumption of a $t$-distribution, the $t$-VaR at level $\alpha$
is, cf.\ Equation~(\ref{eq:5})
\begin{equation}
  \text{VaR}^t_\alpha = -t_{\nu,1-\alpha} \cdot  V_{0} 
  \left(\mathbf w^\intercal\, \boldsymbol{\tilde\Sigma}\, \mathbf
    w\right)^{1/2} = %
  -t_{\nu,1-\alpha} \cdot  V_{0} \cdot  \left(\frac{\nu-2}{\nu}\right)^{1/2}
  \left(\mathbf w^\intercal\, \boldsymbol\Sigma\, \mathbf w\right)^{1/2}.
\end{equation}

Volatility stress at the level $\tilde\alpha\in [0,1]$ is introduced
by setting $V$ to the $\tilde\alpha$-quantile $q_{\tilde\alpha}$ of
the $\text{Ig}(1/2\,\nu,1/2\,\nu)$ distribution. This conveniently
captures that the volatility stress induced is a systematic
event. Furthermore, the severity of the stress event depends on the
heaviness of the tails, expressed by $\nu$. The VaR in this scenario
is determined from
\begin{equation}
  \label{eq:12}
  \p\left(\mathbf w^T X \leq \text{VaR}_{\alpha}|V=q_{\tilde\alpha}\right) %
  = \p\left(\sqrt{q_{\tilde\alpha}}\mathbf w^T A\mathbf Z \leq
    \text{VaR}_{\alpha}\right), %
\end{equation}
with
$\mathbf w^T A\mathbf Z\sim \Ncdf(0,\mathbf w^T
\mathbf{\tilde\Sigma}\mathbf w)$. Consequently, the stressed $t$-VaR
is derived from Equation~(\ref{eq:12}) as a normal distribution VaR
  with the standard deviation scaled according to the fixed mixing
  variable contribution: 
\begin{equation}
  \label{eq:11}
  \text{VaR}^t_{\alpha,\tilde\alpha} = -\Ncdf_{1-\alpha} \cdot V_0\cdot
  \sqrt{q_{\tilde\alpha}} (\mathbf w^T \mathbf{\tilde\Sigma}\mathbf
  w)^{1/2} %
  = -\Ncdf_{1-\alpha} \cdot V_0\cdot
 \sqrt{ q_{\tilde\alpha}}
 \left(\frac{\nu-2}{\nu}\right)^{1/2}  (\mathbf w^T 
  \mathbf{\Sigma}\mathbf w)^{1/2}.  
\end{equation}
To achieve a joint volatility and correlation stress, both methods are
combined: a scaling factor determined from the quantile of the mixing
variable as in Equation~(\ref{eq:11}) is applied independently of a
correlation scenario $\Delta\beta$ as in Proposition
\ref{prop:correlationstress}.

\subsection{Stress test scenario selection}
\subsubsection{Mahalanobis distance}
\label{sec:mahalanobis-distance}

When stress testing, aside from understanding the impact of given
scenarios, one is also interested in the converse question: What is
the worst scenario amongst all scenarios that occur within some
pre-given range?  One way to specify the range is via the so-called
Mahalanobis distance, which measures the distance of a realisation of
a normally distributed random vector from its mean.

Recall that correlations $c_{i,j}$ are modelled as
\begin{align*}
  c_{ij} &= \exp\left(-(\beta_1 |x^1_i-x^1_j| + \beta_2 |x_i^2 - x^2_j| +
           \cdots + \beta_k |x^m_i-x^m_j|) \right), \quad i,j=1,\ldots, n,
\end{align*}
with positive parameters $\beta_1,\ldots, \beta_m$.  If
$\boldsymbol \beta= (\beta_1, \ldots, \beta_m)^\intercal$ is a random
vector with $\E(\boldsymbol\beta)=\boldsymbol{\overline\beta}$ and
covariance matrix $\boldsymbol\Sigma_{\boldsymbol\beta}$, then the
{\em Mahalanobis distance\/} is defined as
\begin{equation*}
  D(\boldsymbol{\beta}) 
  =\left( (\boldsymbol{\beta} - \boldsymbol{\overline\beta})^\intercal
    \boldsymbol{\Sigma}^{-1}_{\boldsymbol\beta} 
    (\boldsymbol{\beta} - \boldsymbol{\overline\beta})\right)^{1/2}.
\end{equation*}
Furthermore, if
$\boldsymbol\beta\sim\Ncdf(\boldsymbol{\overline\beta},
\mathbf\Sigma_{\boldsymbol\beta})$, then the square of the Mahalanobis
distance follows a chi-squared distribution, i.e.,
$D^2(\boldsymbol\beta) \sim \chi^2(m)$.\footnote{The approach can
  easily be extended to heavy-tailed distributions, by assuming that
  $\boldsymbol \beta$ follows an elliptic distribution.}

We are interested in identifying the worst-case scenario
$\boldsymbol\beta^\ast$ that maximises VaR subject to a constraint on
the Mahalanobis distance:
\begin{equation*}
  \boldsymbol\beta^\ast = \argmax_{\boldsymbol\beta:
    D^2(\boldsymbol\beta)\leq h} \text{VaR}_\alpha(\boldsymbol\beta), 
\end{equation*}
where $\text{VaR}_\alpha$ is given by Equation (\ref{eq:5}) with
correlation matrix imposed by $\boldsymbol\beta$. If the parameter $h$
in the constraint is chosen as the $\alpha^\ast$-quantile of the
$\chi^2(m)$-distribution, then $\boldsymbol\beta^\star$ expresses the
worst correlation scenario amongst all scenarios that lie on the inner
ellipsoids covering a probability of $\alpha^\ast$.  From
Equation~(\ref{eq:5}) it is obvious that maximising
$\text{Var}_\alpha$ does not depend on $\alpha$ and is equivalent to
maximising the variance. A trivial consequence is that
$\boldsymbol\beta^\ast$ also maximises expected shortfall
$\displaystyle \text{ES}_\alpha = \frac{1}{1-\alpha} \int_\alpha^1
\text{VaR}_u\, \dd u$. Writing the diagonal matrix with the standard
deviations as the entries on the diagonal as
$\boldsymbol\sigma=(\text{diag}(\boldsymbol\Sigma(\beta)))^{\frac{1}{2}}$,
gives
\begin{equation*}
  \boldsymbol\beta^\ast = \argmax_{\boldsymbol\beta:
    D^2(\boldsymbol\beta)\leq h}\mathbf w^\intercal
  \boldsymbol\Sigma(\boldsymbol\beta)\mathbf w %
  = \argmax_{\boldsymbol\beta:
    D^2(\boldsymbol\beta)\leq h}\mathbf w^\intercal
  (\boldsymbol\sigma\, \boldsymbol C(\boldsymbol\beta) \,\boldsymbol\sigma)
  \mathbf w 
  = \argmax_{\boldsymbol\beta: D^2(\boldsymbol\beta)\leq h}
  \sum_{i=1}^n \sum_{j=1}^n w_i\, w_j\, \sigma_i\, \sigma_j\,
  c_{ij}(\boldsymbol\beta). 
\end{equation*}
The Lagrangian is
\begin{align*}
  \mathcal L &= \boldsymbol w^\intercal(\boldsymbol\sigma\,
               \boldsymbol C(\beta) \,\boldsymbol\sigma) 
               \mathbf w +  \lambda((\boldsymbol{\beta} -
               \boldsymbol{\overline\beta})^\intercal
               \boldsymbol{\Sigma}^{-1}_{\boldsymbol\beta}  
               (\boldsymbol{\beta} - \boldsymbol{\overline\beta})-h)\\ 
             &= \sum_{i=1}^n \sum_{j=1}^n w_i\, w_j\, \sigma_i\, \sigma_j\,
               c_{ij}(\boldsymbol\beta) + \lambda \Big(\sum_{i=1}^m
               \sum_{j=1}^m
               (\beta_i-\overline\beta_i)(\beta_j-\overline\beta_j) 
               q_{ij} - h\Big),
\end{align*}
with $q_{ij}$ the entries of
$\boldsymbol{\Sigma}^{-1}_{\boldsymbol\beta}$.

The first-order conditions are
\begin{align}
  \frac{\partial}{\partial \beta_l} \mathcal L
  &= -\sum_{i,j=1}^n w_i w_j \sigma_i\sigma_j \e^{-\sum_{k=1}^m
    \beta_k |x_i^k-x_j^k|} \cdot |x_i^l-x_j^l| %
    + 2\lambda \sum_{j=1}^m (\beta_j-\overline\beta_j)q_{lj} =0, \quad
    l=1,\ldots, m \label{eq:3}\\
  \frac{\partial}{\partial \lambda} \mathcal L
  &= D^2(\boldsymbol\beta)-h=0\label{eq:8}
\end{align}

Assuming that all factors are indicators measuring if a property is
present in both securities or not, i.e.,
$|x_i^l-x_j^l|=\1_{\{x_i^l\not=x_j^l\}}$ gives
\begin{align}
  \frac{\partial}{\partial \beta_l} \mathcal L 
  &= - \e^{-\beta_l} \, \underbrace{\sum_{i,j=1}^n w_i w_j
    \sigma_i\sigma_j \e^{-\sum_{k=1,k\not=l}^m 
    \beta_k \1_{\{x_i^k\not= x_j^k\}}
    } \cdot \1_{\{x_i^l\not= x_j^l\}}}_{=c_{l,1}}\nonumber \\
  &\phantom{=\,} + \underbrace{2\lambda \Big(\sum_{j=1,j\not=l}^m
    (\beta_j-\overline\beta_j)q_{lj}\Big) - 2\lambda \overline\beta_l
    q_{ll}}_{=c_{l,2}}  + 
    \underbrace{2\lambda q_{ll}}_{=c_{l,3}} \beta_l\nonumber \\
  &= -c_{l,1}  \e^{-\beta_l} + c_{l,2} + c_{l,3} \beta_l= 0, \quad
    l=1,\ldots, k\label{eq:6}\\
  \frac{\partial}{\partial \lambda} \mathcal L
  &= D^2(\boldsymbol\beta)-h=0 \label{eq:7}
\end{align}

Assuming throughout that the factors are chosen in such a way that at
least for one pair of securities the respective indicator is $1$
implies that $c_{l,1}\not=0$, for all $l=1,\ldots, k$.

\begin{proposition}
  The solutions to (\ref{eq:6}) satisfy
  \begin{equation*}
    \beta_l^\ast = W\left(\frac{c_{l,1} \e^{c_{l,2}/c_{l,3}}}
      {c_{l,3}}\right) - \frac{c_{l,2}} {c_{l,3}},\quad l=1,\ldots, k,
  \end{equation*}
  where $W(z)$ is the {\em Lambert $W$-function} (also called {\em
    product logarithm}), which gives the solution for $w$ in
  $z=w\, \e^w$, $z\in \mathbb C$.
\end{proposition}

\begin{proof}
  For ease of notation, we omit the index $l$, so we show that
  $-c_1\, \e^{-\beta} + c_2 + c_3 \beta=0$, with
  $\beta=\displaystyle W\left(\frac{c_1 \e^{c_2/c_3}} {c_3}\right) -
  \frac{c_2}{c_3}$. Setting
  $w:=\displaystyle W\left(\frac{c_1 \e^{c_2/c_3}} {c_3}\right)$ gives
  \begin{equation*}
    -c_1\, \e^{-(w-c_2/c_3)} + c_2 + c_3 (w-c_2/c_3)= -c_1\,
    \e^{-(w-c_2/c_3)} + c_3 w = 0,
  \end{equation*}
  which can be re-arranged to
  \begin{equation*}
    -\frac{1}{w} \e^{-w}\, c_1 \e^{c_2/c_3} + c_3=0. 
  \end{equation*}
  Using that
  $\displaystyle\frac{1}{w} \e^{-w} = \frac{c_3}{c_1\, \e^{c_2/c_3}}$
  yields the claim.
\end{proof}

\subsubsection{Homogeneous portfolio analysis}
\label{sec:homog-portf-analys}

To better understand the stress testing effect we consider a stylised,
homogeneous portfolio as in Section \ref{sec:stress-test-homog} and
determine the worst stress scenario that lies within a pre-specified
Mahalanobis distance. As before, the $m$ risk factors are binary and
the number of securities is $n=2^m$ comprising all $2^m$ risk factor
combinations. The securities all have equal volatility, and the
portfolio is equally-weighted. The risk factor coefficients
$\boldsymbol\beta$ are also assumed to be homogeneous, i.e., they have
identical means $\overline\beta$, variances $\sigma_\beta^2$ and
correlations $\rho_\beta$.

\begin{proposition}
  In the homogeneous setting, the risk factor coefficients of the
  worst scenario within a given Mahalanobis distance $\sqrt{h}$ are
  constant, i.e., $\beta_1^\ast = \cdots = \beta_m^\ast=\beta^\ast$,
  and given by
  \begin{equation*}
    \beta^\ast=\overline\beta -\sqrt{\frac{h \sigma_\beta^2
        (1+(m-1)\rho_\beta)} {m}}. 
  \end{equation*}
\end{proposition}

\begin{proof}
  Because of the binary risk factors, the first-order conditions
  (\ref{eq:3}) simplify to
  \begin{equation*}
    \frac{\partial}{\partial \beta_l} \mathcal L = %
    -\frac{\sigma^2}{n^2} \sum_{i,j=1}^n \e^{-\sum_{k=1}^m \beta_k
      \1_{\{x_i^k\not=x_j^k\}}} \cdot \1_{\{x_i^l\not=x_j^l\}}  +
    2\lambda \sum_{k=1}^m (\beta_k-\overline\beta) q_{lk} =0,
    \quad l=1,\ldots, m.%
  \end{equation*}
  where $q_{lk}$ are the entries of
  $\mathbf \Sigma_{\mathbf \beta}^{-1}$ and, because of the
  homogeneity of $\boldsymbol\beta$, $q_{11}=\cdots=q_{mm}$ and
  $q_{lk}$ constant for all $l\not=k$. It is easily verified that
  $q_{11} = \displaystyle \frac{(m-2)\rho_\beta+1}{(1+(m-2)\rho_\beta
    - (m-1)\rho_\beta^2)\sigma_\beta^2}$ and
  $q_{12} = \displaystyle -\frac{\rho_\beta}{(1+(m-2)\rho_\beta -
    (m-1)\rho_\beta^2)\sigma_\beta^2}$.

  For fixed $l$, the number of instances where
  $\1_{\{x_i^l\not=x_j^l\}}=1$, $i,j=1,\ldots,n$, is $n^2/2$; in
  particular, this number is constant regardless of the choice of
  $l$. Whenever $\1_{\{x_i^l\not=x_j^l\}}=1$, then, across all $i,j$,
  the number of terms where $\1_{\{x_i^k\not=x_j^k\}}=1$ holds is
  equally distributed: the $2^{m-1}$ terms, when $k\not=l$, are the
  result of all combinations of $m-1$ zeros and ones. As a
  consequence, the sums in the first-order conditions have the same
  number of terms and would differ only in $\beta_1,\ldots, \beta_m$;
  however, because they all have the same structure it follows that
  $\beta_1^\ast=\cdots=\beta_m^\ast=\beta^\ast$. Hence, the
  first-order conditions reduce to one condition, which is given by
  \begin{equation*}
    \frac{\partial}{\partial \beta} \mathcal L = -\frac{\sigma^2}{n}
    \sum_{k=0}^{m-1} \binom{m-1}{k} \e^{-\beta\cdot (1+k)} +
    2\lambda (\beta-\overline\beta) (q_{11} + (m-1) q_{12}) = 0.
  \end{equation*}
  Because all $\beta$'s are equal, the worst stress scenario at a
  given Mahalanobis distance $\sqrt{h}$ is one of the two solutions of
  the quadratic equation
  \begin{equation*}
    (\boldsymbol{\beta} -
    \boldsymbol{\overline\beta})^\intercal
    \boldsymbol{\Sigma}^{-1}_{\boldsymbol\beta}  
    (\boldsymbol{\beta} - \boldsymbol{\overline\beta}) =
    \Big(\sum_{i=1}^m \sum_{j=1}^m
    (\beta_i-\overline\beta_i)(\beta_j-\overline\beta_j) q_{ij} \Big) = %
    (m q_{11} + m(m-1) q_{12}) (\beta-\overline\beta)^2 = h. 
  \end{equation*}
  Solving for $\beta$ gives
  \begin{equation*}
    \beta = \overline\beta \pm \sqrt{\frac{h}{m q_{11} + m(m-1)
        q_{12}}} = \overline\beta \pm 
    \sqrt{\frac{h \sigma_\beta^2 (1+(m-1)\rho_\beta)} {m}}. 
  \end{equation*}
  The claim follows because the portfolio variance is monotone
  decreasing in $\beta$.
\end{proof}

Obviously, ceteris paribus, the portfolio risk and VaR of the
worst-case scenario increase with the risk factor variance
$\sigma_\beta^2$ and the risk factor correlation $\rho_\beta$. They
decrease with increasing number of risk factors, $m$. However, we will
see in the examples below that if the initial $\overline\beta$ is
fitted from a given constant asset correlation matrix, then the
worst-case scenario may increase with the number of risk factors as
well.

In this setting, we consider a portfolio where the asset returns have
an average correlation of $0.3$. With five $\beta$ risk factor
coefficients, this is achieved by $\beta=0.5204$ (cf.\ Equation
(\ref{eq:10})). We set $\rho_\beta=0.1972$ and $\sigma_\beta=0.1428$
(these values correspond to the historical averages from the ``London
Whale'' case described below). The $95\%$ worst-case scenario is
$\beta=0.2361$. With an annualised asset volatility of $0.25$, the
initial one-day $99\%$-VaR of $2.09\%$ increases by $33\%$ to
$2.79\%$. Figure \ref{fig:maha} shows the worst-case portfolio
variance increase as a function of the number of correlation risk
factors $m$ as well as for several parameter constellations.  The
bottom left graph of Figure~\ref{fig:maha} shows the impact of a joint
correlation and volatility stress scenario on a one-day $99\%$-VaR. As
laid out in Section~\ref{sec:stress-test-corr}, in addition to the
correlation stress scenario, volatility is scaled to a stress level
corresponding to the $\tilde\alpha$-quantile of a Student
$t$-distribution, cf.\ Equation~(\ref{eq:11}), where we have chosen
$\tilde\alpha = \alpha = 0.99$.  The volatility stress alone increases the
unstressed $t$-VaR by up to 51\%, depending on the parameter $\nu$ of
the $t$-distribution. Stressing both correlation and volatility can
add up to 102\% to the unstressed $t$-VaR.

\begin{figure}[t] 
  \centering \includegraphics[scale=.73]{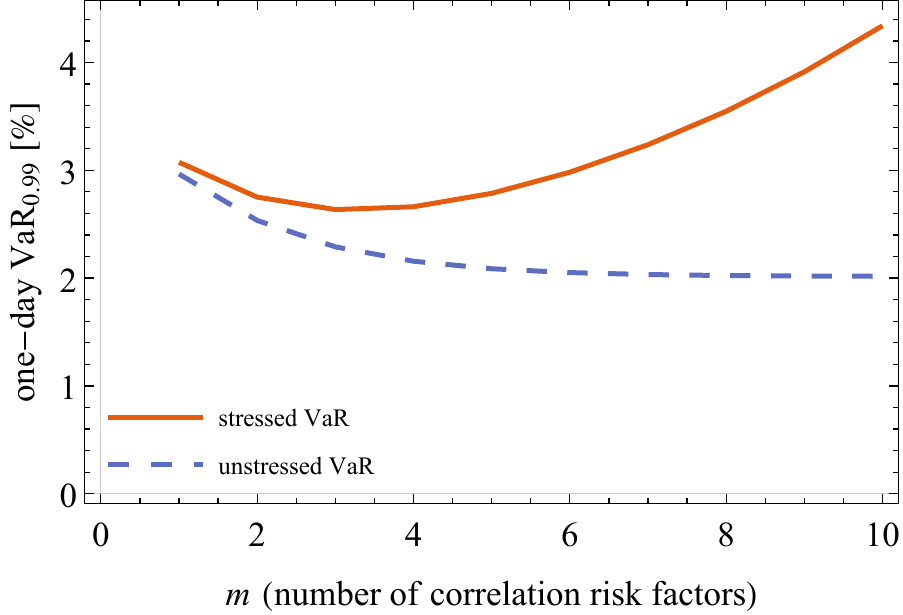}\ \ \ \ \ \
  \includegraphics[scale=.75]{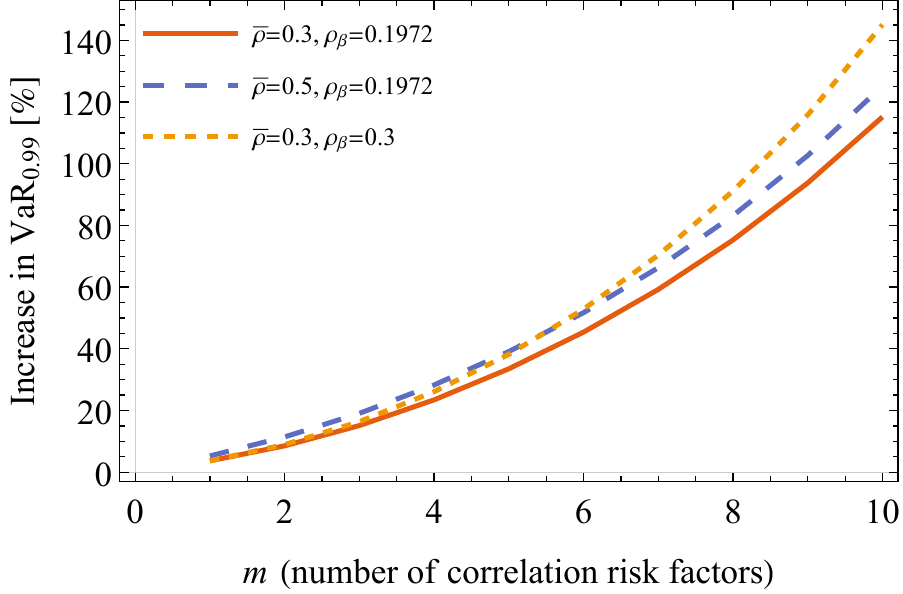}\\\vspace{.25cm}
  
  \includegraphics[scale=.75]{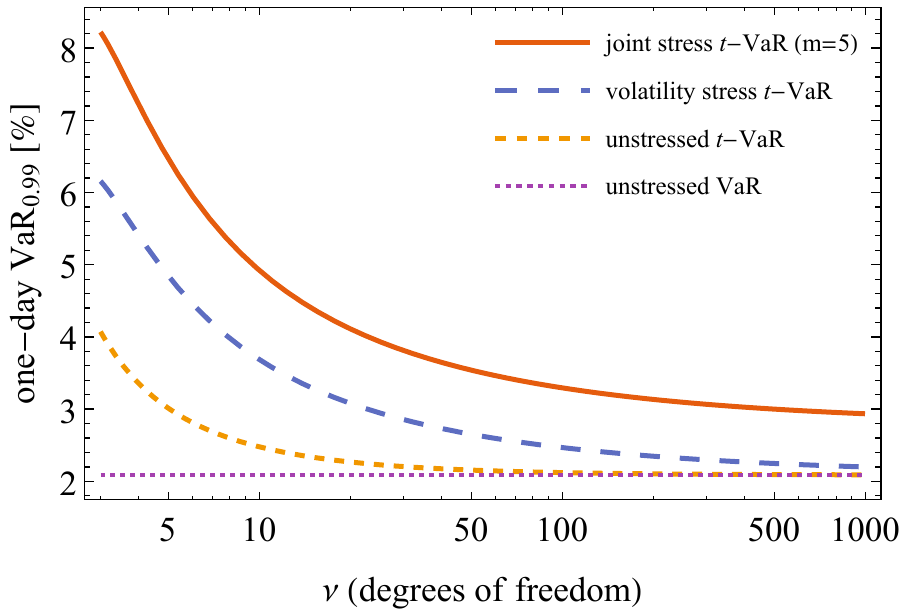}\ \ \ \ \ \
  \includegraphics[scale=.75]{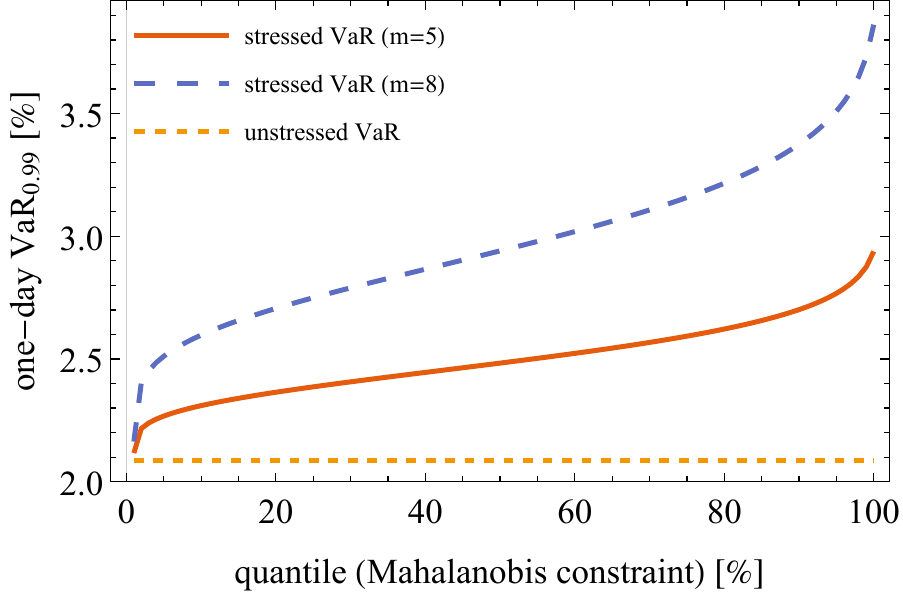}
  \caption{Worst-case (within $95\%$ probability deviation from mean)
    portfolio value-at-risk (VaR). Top left: Initial and stressed
    one-day VaR ($\alpha=0.99$) with average asset correlation
    $\overline\rho=0.3$, annualised asset volatility $\sigma=0.25$,
    beta coefficient correlation $\rho_\beta=0.1972$ and beta standard
    deviation $\sigma_\beta=0.1428$. Top right: Percentage increase in
    VaR for various parameter setups. Bottom left: Joint correlation
    and volatility stress test as a function of $\nu$, with volatility
    stress level $\tilde{\alpha}= \alpha = 0.99$. Bottom right: VaR as a
    function of the Mahalanobis constraint quantile.}
  \label{fig:maha}
\end{figure}

\section{Application to the ``London Whale'' portfolio}
\label{sec:applicationtowhale}

\subsection{The ``London Whale'' case}
In 2012, JPMorgan Chase \& Co.\ reported a loss of approximately
USD~6.2 billion on a credit derivative portfolio that originated on
the books of the comparably small Chief Investment Office (CIO) in
London. This case -- known as the ``London Whale'' -- was generated by
an authorised trading position, so, contrary to most other large
trading losses, it cannot be attributed to fraud or unauthorised
trading.  Interestingly, the loss took place at one of the world's
largest investment banks, widely known for its advanced risk
management, e.g.\ as the the innovator of the widely recognised
RiskMetrics and CreditMetrics frameworks \citep{JPMorgan2013}.

To understand JPMorgan's strategy, trading and risk management of the
loss generating credit portfolio, we consolidate publicly available
information on the London Whale. This section presents our findings in
a very concise format, a detailed review is available at \ifANONYMOUS
[anonymous version: Supplement - Whale Story.pdf].  \else\\
\href{https://ssrn.com/abstract=3210536}{https://ssrn.com/abstract=3210536}.
\fi

JPMorgan, in its function as a lender, is naturally exposed to credit
risk.  In mid-2011, JPM's Chief Investment Office decided to establish
a short credit position via its synthetic credit portfolio (SCP), a
portfolio of credit index derivatives.\footnote{%
  A short position in credit risk corresponds to buying CDS
  protection, i.e., the protection buyer receives default insurance in
  return for a fixed premium. Thus, a deteriorating credit quality
  benefits the protection buyer as a payout becomes more likely and
  hence, the position can be called a short credit risk position.} The
initial purpose of the portfolio was to act as a macro hedge that
would offset naturally long credit exposure
\citep[p. 26]{JPMorgan2013}. A similar strategy and portfolio was
already successfully employed during the 2008--2009 credit crisis. The
decision to re-establish the portfolio was possibly influenced by the
deteriorating credit environment in Europe at that time.

The portfolio was based on the two major global credit derivative
index families, the CDX for the United States and the iTraxx for
Europe.\footnote{The CDX and iTraxx index families are owned, managed,
  compiled and published by Markit Group Limited.
} In addition to the indices, each comprising a portfolio of 125
single-name credit default swaps (CDS), there exists a market of
tranche products, similar to synthetic collateralized debt
obligations, with the indices as underlyings. Both, the CDX and iTraxx
provide different sub-indices, such as an Investment Grade index (IG)
and a High Yield index (HY). At some point, the SCP comprised more
than 120 positions, including most of the active indices and
tranches. For details on the valuation of credit indices and their
tranche products, we refer to Appendix~\ref{sec:tranche-spre-calc} and
\citet{OKane2008}.

The proposed trading strategy was called ``Smart Short''
\citep[p. 51]{USS2013Report}, which translates into a long-short
strategy where credit protection on high yield indices is financed by
selling protection on investment grade indices.  Hence, the upfront
and flow payments can be netted while the resulting portfolio is
sensitive to changes in the market spread between the two position
sides.

By the end of 2011, JPMorgan's senior management assessed an
improvement of the global credit environment, thus requiring less
default protection. Hence the decision was made to reduce the SCP's
risk weighted assets (RWA). The traders in charge estimated that a
direct liquidation of their positions would cost up to USD~590 million
(see internal meeting documents in \citet[Exhibit
8]{USS2013Exhibits}). Faced by this number, the CIO management decided
against a direct reduction, in favour of ``managing'' profit and
losses (P\&L) while gradually reducing RWA over time \citep[pp. 29
ff.]{JPMorgan2013}.

A new trading strategy aimed at reducing RWA by increasing positions
with opposite market sensitivity.  In order to comply with stress
limits, the strategy was implemented by forward spread trades, as was
stated later during an interview with JPMorgan's internal task force
\citep[p. 52]{USS2013Report}. In the context of the SCP, forward
spread trades meant buying protection on short maturity indices, while
selling protection on longer maturities. This would hedge in the near
term but generate credit exposure on the long term.

By the end of January 2012, after experiencing a loss of 50 million
from the default of Kodak, the traders where faced with three
objectives: stemming the year-to-date (YTD) losses on the SCP,
reducing RWA and maintaining protection to prevent default losses
(``Kodak moments''). All objectives were addressed simultaneously by
adding more positions to the portfolio, namely, long risk positions to
participate in the upward moving market, while generating carry to
fund the YTD losses and short risk positions. Additionally, protection
was bought to create positive P\&L from Kodak type events. Therefore,
the traders increased the size of both their long and short positions.

On March 23, 2012, the CIO's most senior executive ordered the traders
to ``put the phones down'', i.e., to cease all related trading
activities \citep[Exhibit 1i]{USS2013Exhibits}. At this point, the SCP
had a net notional of about USD~157 billion \citep[Exhibit
1a]{USS2013Exhibits}, which was 260\% up from the September 2011 net
notional (and slightly more than Vietnam's 2012 GDP). The SCP's top 10
positions as of March 23, 2012 are shown in
Table~\ref{tab:Top10Positions}. Ceasing to trade meant, of course,
that the traders could no longer influence P\&L, and as a consequence
the losses on the SCP sky-rocketed.

\begin{table}
  \centering
  \caption{Top 10 Positions of the SCP as per March 23, 2012, reported
    in USD net notional and as percentage share of the respective
    market. The market's net notional is the net protection bought on
    an index series by net buyers (or equivalently sold by net
    sellers) \citep{DTCC2011}. The publicly available data is
    aggregated (net notional) and may not contain all live positions
    due to possible disclosure restrictions, which explains the
    occurrence of values in excess of 100\%.  }
  \label{tab:Top10Positions}
  \begin{adjustbox}{max width=\textwidth}
    
    \begin{tabular}{lrrrlrr}
      \toprule
      \multicolumn{4}{c}{Index} \\
      \cmidrule(r){1-4}
      Name      & Series & Tenor & Tranche (\%) & Protection & Net Notional (\$) & Share (\%)\\
      \midrule                                                                                   
      CDX.IG    & 9      & 10yr   & Untranched   & Seller      & 72,772,508,000    & 50.19      \\
                & 9      & 7yr    & Untranched   & Seller      & 32,783,985,000    & 22.61      \\
                & 9      & 5yr    & Untranched   & Buyer     & 31,675,380,000    & 21.85      \\
      iTraxx.EU & 9      & 5yr    & Untranched   & Seller      & 23,944,939,583    & 37.01      \\
                & 9      & 10yr   & 22 -- 100     & Seller      & 21,083,785,713    & 22.04      \\
                & 16     & 5yr    & Untranched   & Seller      & 19,220,289,557    & 64.18      \\
      CDX.IG    & 16     & 5yr    & Untranched   & Buyer     & 18,478,750,000    & 78.92      \\
                & 9      & 10yr   & 30 -- 100     & Seller      & 18,132,248,430    & 50.35      \\
                & 15     & 5yr    & Untranched   & Buyer     & 17,520,500,000    & 117.01     \\
      iTraxx.EU & 9      & 10yr   & Untranched   & Seller      & 17,254,807,398    & 26.67      \\
      
      \midrule 
      \multicolumn{1}{l}{Net Total} & \multicolumn{5}{r}{137,517,933,681} \\ 
      \bottomrule
      \multicolumn{7}{l}{ 
      \begin{footnotesize}
        Data source: \citet[Exhibit 36]{USS2013Exhibits} and
        \citet[Section 1, Table 7]{DTCC2014}.
      \end{footnotesize}
      }\\ 
    \end{tabular}
  \end{adjustbox}
\end{table}

Publicly available reports by JPMorgan's internal task force
(\citeyear{JPMorgan2013}) and the United-States-Senate
(\citeyear{USS2013Report}) focus on management and organizational
problems, position misreporting, market manipulation, and
spreadsheet-errors. This neglects that the classical risk measures
employed might have been insufficient in their own right. To monitor
the SCP, JPMorgan (\citeyear{JPMorgan2013}) primarily used the
Value-at-Risk that would be reported in its 10-K
filings. \citet{Cont2015} find that this risk measure was insufficient
due to the size of the SCP, as it scales linearly with position size
and neglects market impact. In addition, the authors state that a
correlation decay, which was observable before the collapse of the
portfolio, was possibly caused by the SCP's own market impact.

Aside from VaR, \emph{credit spread widening of 10\%} (CSW-10) is the
second pivotal risk measure, a sensitivity measure for the profit and
loss impact of a simultaneous 10\% increase in credit
spreads. JPMorgan's traders relied heavily on this measure to balance
their portfolio in a way that offsetting positions would minimize the
overall CSW-10 \citep{JPMorgan2013}.

However, a hedging strategy primarily based on a sensitivity measure,
CSW-10 in this case, ignores that correlation amongst the portfolio
components may be imperfect. Value-at-risk takes into account
correlations, but to the best of our knowledge, potential {\em
  changes\/} in correlation were ignored. As the SCP was a portfolio
composed of a large number of offsetting positions that are highly,
but not perfectly dependent, correlation is easily seen to be a, if
not {\em the}, crucial risk driver. A change in the correlations, for
instance amongst high-yield and investment grade positions, amongst
index and tranche positions, amongst CDX and itraxx positions, could
easily lead to large P\&L swings.

The portfolio hedging alone is a strong indicator for the correlation
dependence of the SCP. Additionally, the size of the SCP and the
resulting market-impact could have affected correlation. It is
therefore possible that on top of the normal variation in correlation,
the portfolio was exposed to more erratic changes that would be
captured only by stress tests.

\subsection{Correlation stress testing the ``London Whale''}
\label{sec:data}

\subsubsection{Correlation methodology}
\label{sec:whale-method}

In the following, the sensitivity to various correlation scenarios of
the SCP position is calculated from historical data. The analysis is
based on the portfolio composition of 23 March 2012, the day when
trading ceased. The historical data are provided by Markit and consist
of daily CDX and itraxx spreads and tranche data (spreads, upfront
payments and base correlations) of the series in place. Details on how
the tranche data was transformed to credit spreads are given in
Appendix~\ref{sec:tranche-spre-calc}.

\begin{figure}[t]
  \centering
  \includegraphics[scale=.435]{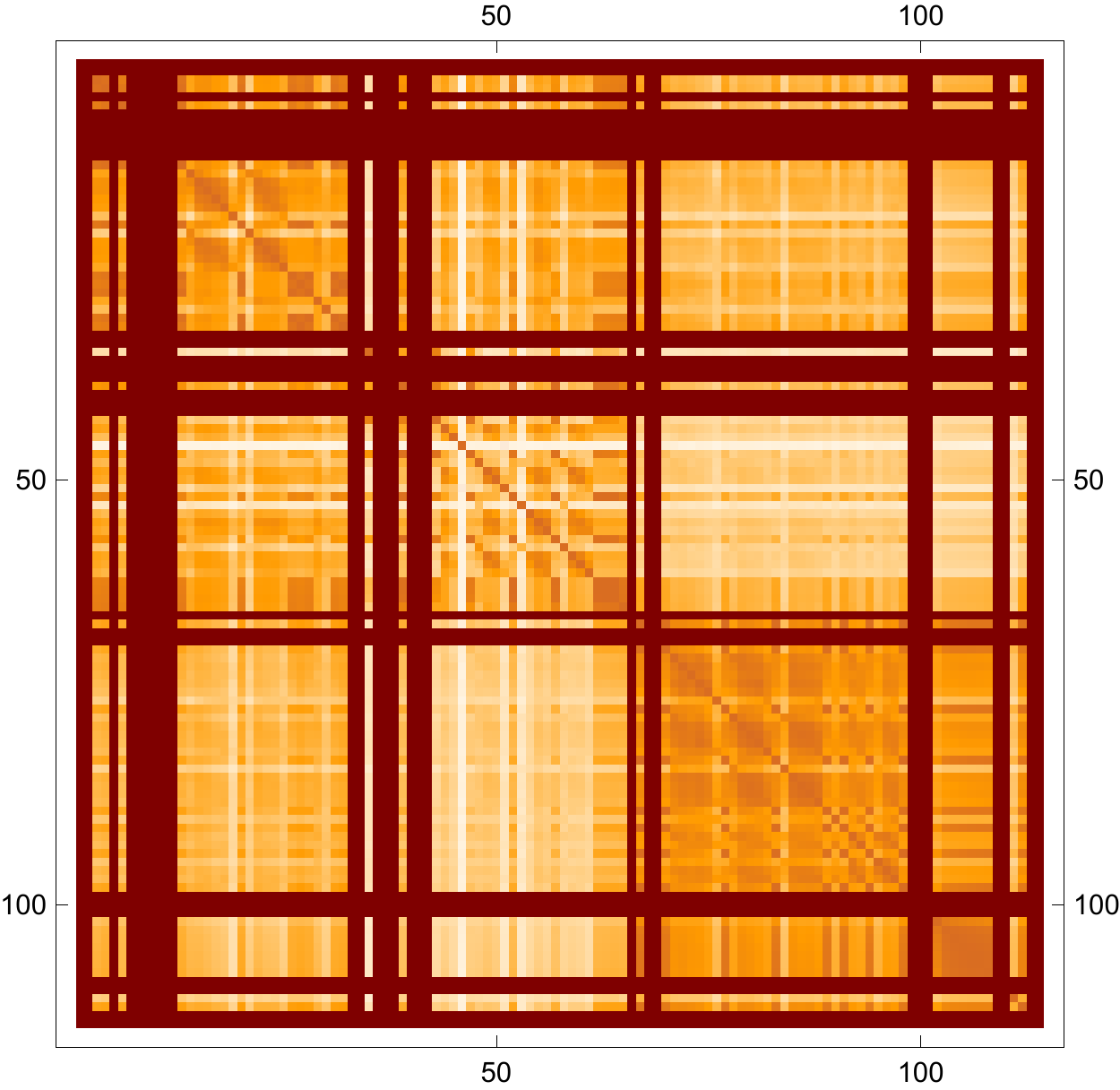}\ \ \ \ \ \ \
  \includegraphics[scale=.6]{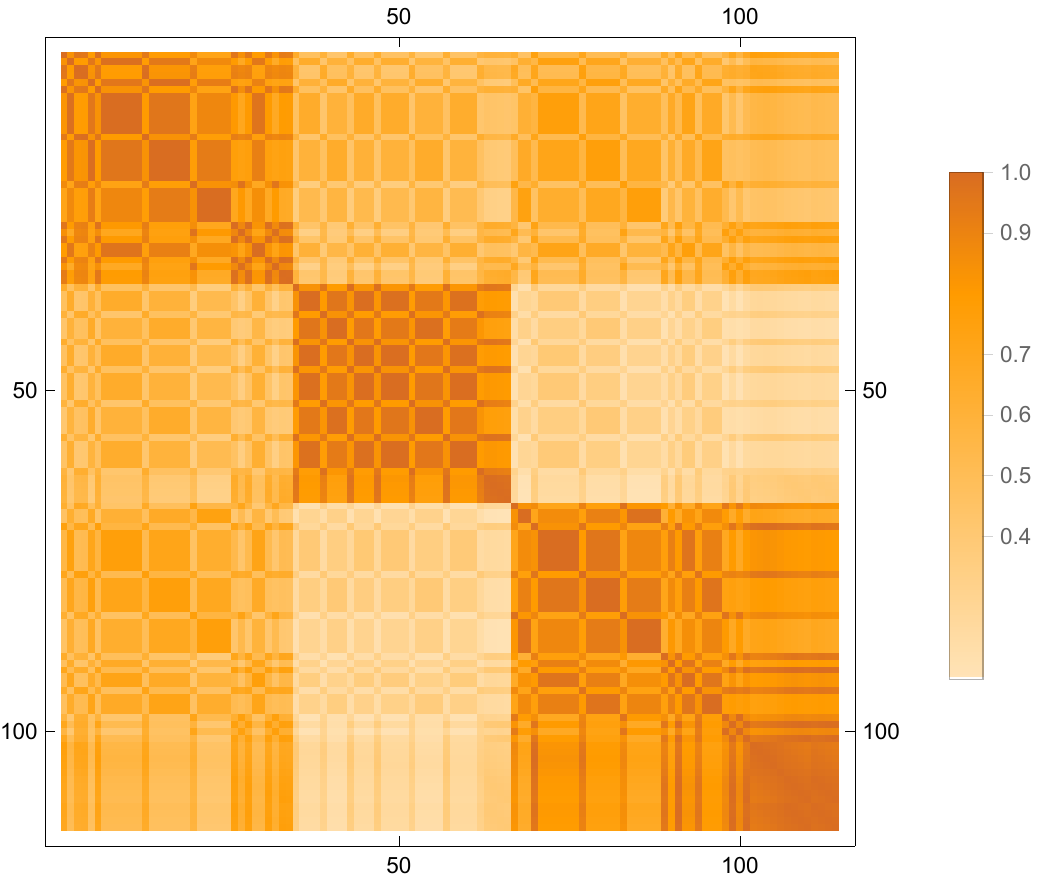}
  \caption{Correlation matrices of 23 March 2012. Left: Empirical
    correlation matrix; right: parameterised correlation matrix. The
    dark red entries are unavailable correlations due to insufficient
    data. The three blocks of highly correlated data consist of (from
    top to bottom): CDX IG, CDX HY and iTraxx securities.  }
  \label{fig:corr_matrices}
\end{figure}

As the main risk factors affecting correlation the following five
properties are identified: maturity, index series, investment grade
(yes/no), CDX vs.\ itraxx, index vs.\ tranche. Information about
seniority of tranches was considered, but failed to provide useful
results. Hence, the correlation $c_{ij}$ of credit spread returns of
credit derivatives indexed by $i$ and $j$ is given by:
\providecommand{\isCDX}{\text{isCDX}}
\providecommand{\isIG}{\text{isIG}}
\providecommand{\maturity}{\text{maturity}}
\providecommand{\series}{\text{series}}
\providecommand{\isIndex}{\text{isIndex}}
\begin{multline}
  c_{ij} = \exp\big(-(\beta_1 |\isCDX_i - \isCDX_j| + \beta_2
  |\isIG_i-\isIG_j| + \beta_3 |\maturity_i-\maturity_j| \big.\\
  \big.+ \beta_4|\series_i - \series_j| + \beta_5|\isIndex_i -
  \isIndex_j|)\big). \label{eq:1}
\end{multline}
In the results provided below, all distance measures are normalised to
$[0,1]$, which makes the impact of the calibrated parameters
comparable.

At any point in time $t$, the parameters $\beta_1, \ldots, \beta_5$
are calibrated from the $250$ credit spread returns preceding day
$t$. Daily parameters are calibrated starting from 1 March 2011
through 12 April 2012. The instruments entering the calibration are
the $117$ instruments identified to be in the SCP portfolio on 23
March 2012. The precise set of instruments entering on each date
differs slightly through time for various reasons such as maturing
instruments, spread availability, etc. Figure \ref{fig:corr_matrices}
shows the empirical correlations and the calibrated correlations from
Equation (\ref{eq:1}) as of 23 March 2012. The calibrated coefficients
$\boldsymbol{\beta} = (0.35, 0.37, 0.21, 0.05, 0.20)^\intercal$
indicate a strong de-correlation amongst the regional 
property (CDX vs.\ itraxx) and the credit quality (investment grade
vs.\ high-yield), a lesser de-correlation amongst maturity and amongst
index vs.\ tranche product. The series-factor on the other hand
provides a strong correlation.

The calibrated parameters for the whole time period (1 March 2011--12
April 2012) are shown in Figure \ref{fig:corr_parameters}. The chart
shows that the credit quality (investment grade versus high yield)
de-correlated over time, while the correlation differences driven by
the region (CDX vs.\ itraxx) decreased. Especially in Q4 2011 and Q1
2012, when strategic decisions regarding the SCP were made, these were
major drivers of correlation changes.

\begin{figure}[t] 
  \centering
  \includegraphics[scale=1]{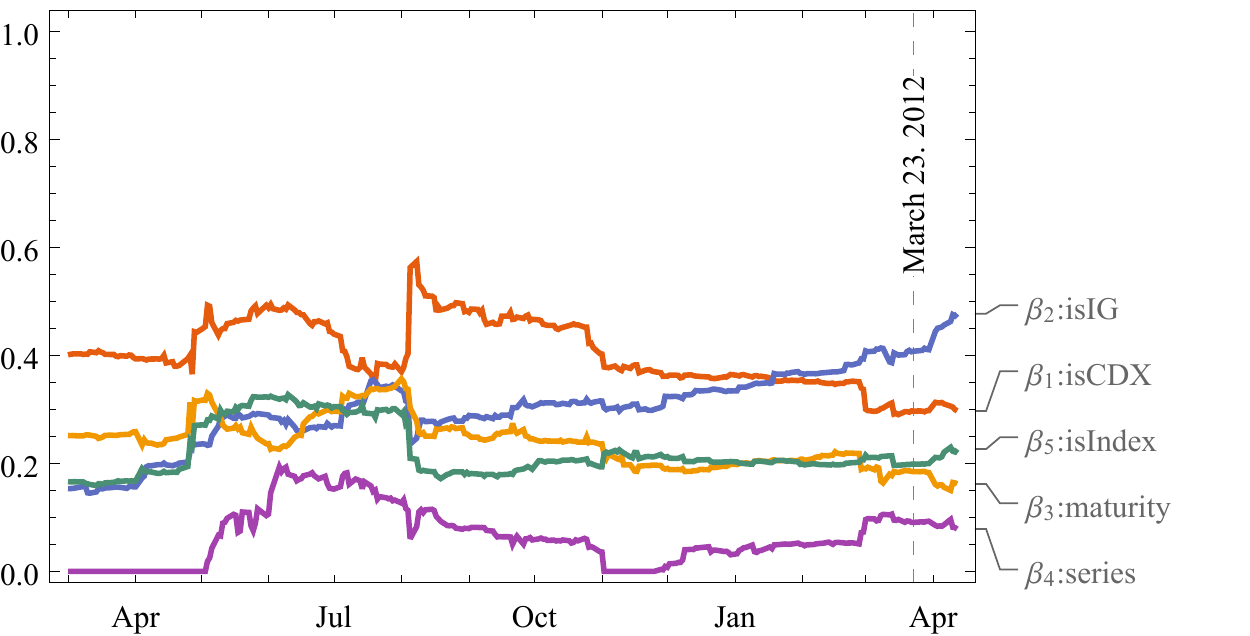} 
  \caption{Coefficients associated with correlation parameterisation
    of CDX and itraxx positions in London Whale position;
    01/03/2011--12/04/2012. All distances are normalised to $[0,1]$ to
    make the coefficients comparable. Data source: Markit.}
  \label{fig:corr_parameters}
\end{figure}

\subsubsection{CDS portfolio risk}
\label{sec:cds-portfolio-risk}

The risk of a CDS portfolio is expressed by value-at-risk (VaR) using
the variance-covariance approach, i.e., the portfolio change is
approximated by a first-order Taylor approximation in the credit
spreads, and credit spread returns are assumed to be normally
distributed. The portfolio risk is then fully captured by the
portfolio variance. To simplify notation, we omit the maturity of a
CDS contract and use the following notation related to CDS position
$i$ in a portfolio of $n$ CDS positions: $s_{i,t}$ denotes the fair
spread at time $t$, $A_i$ is the notional amount of the position and
$\text{RPV01}_{i,t}$ is the risky PV01 at time $t$. Details, such as
the calculation of $\text{RPV01}_{i,t}$, are given in
Appendix~\ref{sec:cds-valuation}.

The portfolio value is then expressed as
$\displaystyle V_t = \sum_{i=1}^n V_{i,t} \approx \sum_i A_i\,
\text{RPV01}_{i,t}\, (s_{i,0}-s_{i,t})$, where $A_i$ is positive for a
short protection position and $A_i$ is negative for a long protection
position.  The portfolio P\&L $\Delta V$ is approximated by spread
returns in the following way:
\begin{align*}
  \Delta V&\approx -\sum_i A_i\, \text{RPV01}_{i,t-1}\, \Delta s_i %
            =- \sum_i A_i\, \text{RPV01}_{i,t-1}\, \frac{\Delta
            s_i}{s_{i,t-1}}\, s_{i,t-1}\\%
          &=\underbrace{- \Big(\sum_j A_j \text{RPV01}_{j,t-1}\,
            s_{j,t-1}\Big)}_{=:V_{t-1}}\cdot
            \sum_i w_i\, r_i,
\end{align*}
where
$w_i=\displaystyle\frac{ A_i\, \text{RPV01}_{i,t-1}\, s_{i,t-1}}
{\sum_j A_j \text{RPV01}_{j,t-1}\, s_{j,t-1}}$ denotes the percentage
weight of the position in the portfolio and
$r_i=\displaystyle\frac{\Delta s_i}{s_{i,t-1}}$ denotes the spread
return. For ease of notation, we write $V_{t-1}$ for the linear
approximation of the portfolio value.

Now, assuming that
$\mathbf r=(r_1,\ldots, r_n)^\intercal \sim
\Ncdf(0,\boldsymbol\Sigma)$, i.e., spread returns are jointly normally
distributed with expectation $0$ (a reasonable assumption for small
time horizons) and covariances described by the $n\times n$ matrix
$\boldsymbol\Sigma$, the portfolio VaR is again given by the
variance-covariance approach, see Equation (\ref{eq:5}). 

\subsubsection{Results}
On March 23, 2012 JPMorgan's senior management ordered to cease all
trading activities for the SCP. The exact portfolio composition is
known only for this day from publicly available sources. To calculate
risk figures for the SCP, the relevant credit index data is taken from
Markit and converted as necessary via the credit valuation model in
Appendix~\ref{sec:cds-valuation}.

The approach uses historical data to fit parameters for both the P\&L
distribution and the correlation model. After processing the data and
excluding constituents with insufficient observations, 93 constituents
with a total net notional of USD~154.34 billion remain to be included
in the calculations. The unstressed delta-normal 1-day VaR at the 99\%
confidence level is USD~339.32 million (base case), which is about
twice as high as the VaR reported by \citet[pp. 124
ff.]{JPMorgan2013}. This number will be used as a benchmark for
scenarios with stressed correlations.

The problem of finding the constrained global maximum of
$\text{VaR}_\alpha(\boldsymbol\beta)$ is solved numerically. To ensure
robustness, first and second order conditions as well as different
algorithms are reviewed, including Nelder Mead, Differential
Evolution, Simulated Annealing and Random Search, all with the same
result. In terms of computational time, Simulated Annealing appears
most efficient.

As laid out in Section~\ref{sec:mahalanobis-distance}, a plausibility
constraint can be applied to the stress testing method. For this
application, where correlation parameters are assumed to follow a
multivariate normal distribution, this means considering only
correlation scenarios that lie on or below a quantile ellipsoid, which
is determined by a Mahalanobis distance. Out of the set of feasible
scenarios, the one with the highest value at risk for a given quantile
is reported in Table~\ref{tab:whaleVaRqs}.

\begin{table}[t]
  \centering
  \caption{The SCP portfolio's 1-day 99\% value-at-risk for different
    Mahalanobis quantile constraints. Percentage changes denote the
    relative distance to the base VaR, i.e., the VaR under the
    original setup of March 23, 2012. The joint stress test captures
    simultaneous changes in correlation and volatility, with
      percentage changes referring to the base $t$-VaR scenario. The
    heaviness of the tails of the return distribution is calibrated to
    $\nu=13.5$. The volatility stress level $\tilde\alpha$ for the
    joint stress test is set to the quantile in column one. }
  \label{tab:whaleVaRqs} 
  \begin{tabular}{lrrrrr}
    \toprule
    & \multicolumn{3}{c}{correlation stress} & \multicolumn{2}{c}{joint stress} \\
    \cmidrule(lr){2-4} \cmidrule(lr){5-6}
    Quantile & VaR$_{0.99}$ & $t$-VaR$_{0.99}$ & Change(\%) & $t$-VaR$_{0.99}$ & Change(\%)\\
    \midrule 
    base case & 339.32 &  354.98 & & 354.98 &\\
    
    0.7   & 366.87 & 383.80 & 8.12  & 386.28 & 8.82 \\
    0.8   & 369.39 & 386.44 & 8.86  & 416.41 & 17.31 \\
    0.9   & 372.89 & 390.10 & 9.89  & 464.40 & 30.83 \\
    0.95  & 375.76 & 393.11 & 10.74 & 510.54 & 43.82 \\
    0.99  & 381.08 & 398.67 & 12.31 & 617.38 & 73.92 \\
    0.995 & 383.00 & 400.68 & 12.87 & 664.73 & 87.26 \\
    0.999 & 386.88 & 404.74 & 14.02 & 780.37 & 119.84 \\

    unconstrained$^*$ & 620.96 & 649.62 & 83.00 & 1252.53 & 252.85\\
    \bottomrule 
    \multicolumn{6}{l}{{\tiny $^*$Unconstrained w.r.t. correlation changes; $\tilde\alpha$ remains on the 0.999 level.}} \\
  \end{tabular}
\end{table}

For the 99\%-quantile constraint the (variance-covariance)-VaR
is USD~381.08 million, which corresponds to a 12.31\% increase
relative to the base case. This is a substantial increase given that a
daily VaR increase at the 99\% level is expected to occur several
times a year. The worst case, which is unconstrained with respect to
the Mahalanobis distance of the parameters, produces a 1-day VaR of
USD~620.96 million, which is 83.01\% greater than the base case. These
results ignore changes from other risk factors, such as volatility,
that would typically increase in a downside scenario as well. In fact,
its size makes the SCP especially vulnerable to other factors, such as
market liquidity issues and resulting plausible correlation scenarios
that are not reflected in historical data.

The results of a joint stress test capturing simultaneous changes in
correlation and volatility are also presented in
Table~\ref{tab:whaleVaRqs}. Here, the volatility stress level
$\tilde\alpha$ is set to the same quantile as the Mahalanobis
  constraint. The parameter $\nu$,
which captures the heaviness of the tails of the underlying return
distribution, is determined by fitting a multivariate $t$-distribution
to the 250 trading days of returns prior to March 23,
2012.\footnote{The parameter $\nu$ is fitted by a combination of
  matching the first two moments and maximum likelihood.}  With
$\nu = 13.5$, the unstressed $t$-VaR is already slightly higher than
the normal VaR, i.e. USD~354.98 million instead of 339.32 million.
The joint stress $t$-VaR is roughly USD 617~million at the 99\%
confidence level, which corresponds to a 73.92\% increase over the
unstressed $t$-VaR.

A correlation stress test would have enabled JPMorgan's risk
management to identify key risk drivers and more appropriately assess
the risks of its portfolio.  Figure~\ref{fig:boxplots_params} shows
box-plots for the correlation parameters as well as the parameters as
of 23 March 2012, and the worst case parameters at a Mahalanobis
constraint equivalent to 99\%. All parameters, with the exception of
$\beta_2$, which identifies whether an index is investment grade or
high-yield, are stressed upwards, hence, decorrelate. The slight
downward shift of $\beta_2$ can be attributed to a parameter increase
prior to the stress test.

Furthermore, \citet{Cont2015} report a breakdown of correlation
between CDX.IG.9 and CDX\allowbreak.IG.10 immediately after trading
was halted in March 2012, when the CIO started to sell its large
positions (see CDX.IG.9 in Table~\ref{tab:Top10Positions}), which ''is
a signature of the market impact of the CIO's trading''. This
structural break could not have been predicted by historical data,
which shows that the actual risk -- owed to the size of the portfolio
-- was closer to the worst case scenario than suggested by the
plausibility constraints. In the unconstrained (worst-)case, the
parameter $\beta_4$, capturing decorrelation between index series,
appears as the major risk driver.

\begin{figure}[t]
  \centering \includegraphics[scale=.7]{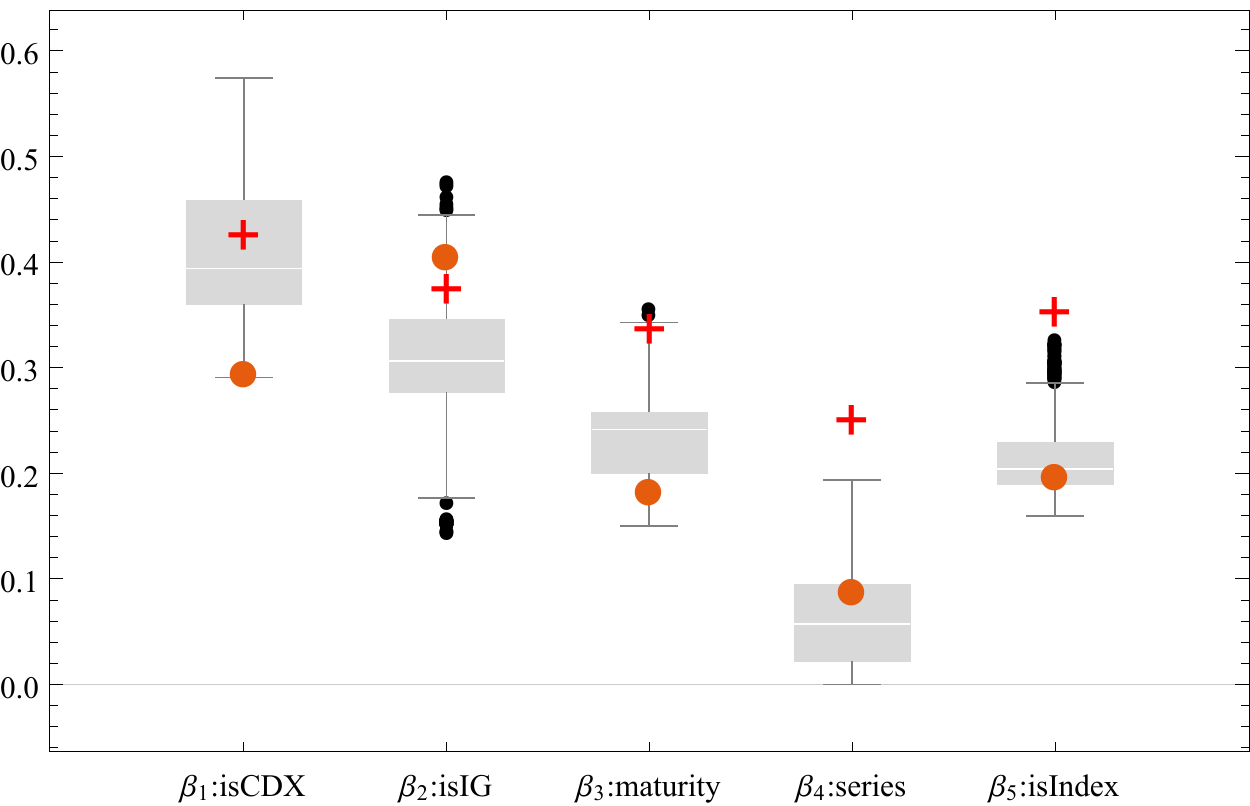}
  \caption{Box-plots of correlation parameters. Dots indicate the
    observed parameters as of 23.03.2012, crosses indicate the
    worst-case scenario under a 99\%-quantile Mahalanobis distance
    constraint.}
  \label{fig:boxplots_params} 
\end{figure}

\section{Concluding remarks and outlook}
\label{sec:conclusion}

The dependence structure amongst portfolio components is of great
relevance to the risk inherent in a financial portfolio, and as such,
correlation stress testing provides important information about
portfolio risk.  The methodology developed in this paper maps risk
factors to correlations, which in turn allows to specify correlation
stress test scenarios in terms of risk-factor changes. In addition,
worst-case risk-factor scenarios can be identified, yielding insights
into the main factors driving portfolio risks and potential losses. We
derive analytical results that allow for computing the value-at-risk
impact both for given scenarios and for determining worst-case
scenarios.

To illustrate the method in a realistic setting, the correlation
stress testing methodology is applied to the case of the ``London
Whale''. This serves as an interesting case, because -- in an attempt
to decrease the portfolio's riskiness -- the notional amount was
increased significantly by adding highly-correlated offsetting
positions. Such a trading strategy is extremely vulnerable to
correlation changes and therefore lends itself to illustrating the
importance of correlation stress testing. Historical data suggests
that, amongst the worst $1\%$ of correlation scenarios, the $1$-day
$99\%$-VaR of the portfolio would have increased by $12\%$ or
more. Such a scenario is expected to occur $2$--$3$ times a year. The
overall worst-case correlation scenario for the portfolio entails a
VaR increase of $83\%$. These are ceteris paribus results, isolating
the effect of a correlation change and neglecting that in reality
large correlation changes would typically occur jointly with
volatility swings. A joint scenario where volatility and correlation
are jointly stressed yields an increase of 73.92\% in the $99\%$
$t$-VaR at a $99\%$-volatility stress leveland a worst-case increase
in VaR of 252.85\%. 

Correlation stress tests are particularly insightful on portfolios
with large positions, as adverse correlation moves may be triggered
from the market impact of large trades. More specifically, the
ordinary co-movement of two or more assets may be disturbed by the
market price impact of large trades. Hence, an appropriate risk
assessment of large portfolios must take into account that a small
risk footprint under conventional risk measures might be due to
offsetting exposures, where only a correlation risk measure can reveal
the true portfolio risk.

A further application of correlation stress testing as developed in
this paper would be the analysis of central counterparties (CCP),
which clear exceptionally large financial portfolios. Currently, 62\%
of the USD~544 trillion in notional outstanding of interest rate
derivatives is cleared through CCPs \citep{Wooldridge2016central}.
As part of their mandate and to protect themselves from the default of
clients, CCPs use a margining system consisting of a {\em variation
  margin\/} (daily mark-to-market settlement) and an {\em initial
  margin\/} (buffer to cover market losses following a client's
default).
To account for diversification benefits and to reduce clients'
clearing costs, margins are typically calculated on a client's netted
position.\footnote{Similar offers are found for almost any financial
  service provider that employs a margining system, e.g.\ brokers and
  future exchanges.} The resulting margin requirements may be highly
correlation-sensitive.  Moreover, adverse correlation scenarios may
affect many or even all clients, creating simultaneous margin calls to
post additional collateral. The correlation stress testing method
developed here is capable of identifying these kinds of systemic risk
events. In addition, correlation stress testing on the clearing client
side would provide insights on possible future margin requirements and
the resulting collateral funding risk.


\appendix

\section{Credit default swap valuation}
\label{sec:cds-valuation}
Given an underlying entity (e.g.\ a sovereign or a company), a {\em
  credit default swap (CDS)} is a contract between two counterparties,
the protection buyer and the protection seller, that insures the
protection buyer against the loss incurred by default of the
underlying entity within a fixed time interval.  The protection buyer
regularly pays a constant premium, the {\em credit spread\/} or {\em
  CDS spread}, which is fixed at inception, up until maturity of the
CDS or the default event, whichever occurs first.  This stream of
payments is termed the {\em premium leg\/} of the CDS.  In return, the
protection seller agrees to compensate the protection buyer for the
loss incurred by default of the underlying entity at the time of
default in case this occurs before maturity.  This constitutes the
{\em protection leg\/} of the CDS. The CDS spread that makes the value
of the premium leg and the protection leg equal is the {\em fair CDS
  spread\/}.

More precisely, let $r>0$ denote the default-free interest rate,
assumed to be constant for simplicity. Furthermore, assume that the
payment at default is a fraction $(1-R)$ of the notational amount,
$R\in [0,1)$.  The probability of default of the underlying entity at
time $t$ until time $T$ is denoted by $P(t,T)$; this probability is
conditional on any information available until time $t$.  Denote by
$s(t,T)$ the fair credit spread at time $t$ of a CDS with maturity
$T$.  Here, we follow the convention that entering into a CDS involves
no initial cash-flow, that is, the market value of a CDS at inception
is $0$. Even though CDS are nowadays traded with an upfront payment,
this still corresponds to the common quoting convention. In other
words, the discounted fair values of the premium and the default legs
are equal. In addition, to simplify the exposition and notation, we
assume that credit spreads are paid continuously instead of quarterly.
From the point of view of a protection seller, the value of a CDS
contract entered at time $t$ is given via risk-neutral pricing by the
value of the premium leg minus the value of the protection
leg. Combining this with the fact that the initial value of the CDS is
$0$ gives:
\begin{equation*}
  0 = s(t,T)\, \underbrace{\int_t^T \e^{-r(u-t)}\, (1-P(t,u))\, \dd
    u}_{=\text{RPV01}} - (1-R)\, \int_t^T \e^{-r(u-t)}\, P(t,\dd u), 
\end{equation*}
or, re-arranging for the spread,
\begin{equation}
  \frac{s(t,T)}{1-R}  = \frac{\int_t^T \e^{-r(u-t)}\, \dd P(t,u)}%
  {\int_t^T \e^{-r(u-t)}\, (1-P(t,u))\, \dd u}. 
  \label{eq:18}
\end{equation}
The term RPV01 denotes the {\em risky present value of a basis
  point}. For a full derivation we refer to e.g.\ Chapter 6 of
\cite{OKane2008}.

The {\em mark-to-market value\/} of an existing CDS position is
expressed as the cost of unwinding the transaction by entering into an
offsetting CDS position. In the following we assume a CDS contract
with maturity $T$ and notional $\$ 1$ entered at time $0\leq t$ from
the point of view of the protection seller.  Conditional on no-default
at time $t$, the value of the position at time $t$ is
\begin{align*}
  V_t &= s(0,T)\, \int_t^T \e^{-r(u-t)}\, (1-P(t,u))\, \dd u - (1-R)\,
        \int_t^T
        \e^{-r(u-t)}\, P(t,\dd u)\\
      &= \left(s(0,T)-s(t,T)\right) \, \int_t^T \e^{-r(u-t)}\,
        (1-P(t,u))\, \dd u %
        = \left(s(0,T)-s(t,T)\right) \, \text{RPV01}(t,T). 
\end{align*}
Here we have used that the values of the premium and protection legs
of the time-$t$ CDS are equal.

A simplification of the valuation occurs by assuming that, similar to
a constant interest rate, the default probabilities are subject to a
constant {\em hazard rate\/} $\lambda_t>0$, i.e.,
$P(t,T) = 1-\e^{-\lambda_t (T-t)}$, $T\geq t$.  The credit spread
formula (\ref{eq:18}) reduces to the so-called {\em credit spread
  triangle},
\begin{equation}
  \label{eq:4}
  \frac{s(t,T)}{1-R} = \lambda_t,
\end{equation}
and the RPV01 is then expressed as
\begin{equation}
  \text{RPV01}(t,T) = \int_t^T \e^{-(r+\lambda_t)(u-t)}\, \dd u %
  = \int_t^T \e^{-(r+s(t,T)/(1-R))(u-t)}\, \dd u. %
  \label{eq:2}
\end{equation}
For value-at-risk calculations, it is useful to approximate the P\&L
$\Delta V$ by a first-order Taylor-approximation on the spread change
$\Delta s=s(t,T)-s(t-1,T)$:\footnote{Assuming that $\Delta t$ is
  small, we ignore the change due to time-decay.}
\begin{align*}
  \Delta V  =V_t-V_{t-1}&\approx \frac{\partial}{\partial s} V_{t-1}\,
                          \cdot \Delta s\\ 
                        &= -\text{RPV01}(t-1,T)\cdot \Delta s + (s(0,T)-s(t-1,T)) \cdot
                          \frac{\partial}{\partial s} \text{RPV01}(t-1,T)\cdot \Delta s. 
\end{align*}
The second term involves a product of spread changes and is therefore
smaller, so we shall ignore it for ease of computations, giving
\begin{equation*}
  \Delta V \approx -\text{RPV01}(t-1,T)\cdot \Delta s.
\end{equation*}

\section{Tranche spread calculation}
\label{sec:tranche-spre-calc}

This appendix gives a brief outline of the calculation of fair spreads
of credit index tranches required for estimating the
$\beta$-parameters of the ``London Whale'' portfolio. The calculations
are based on \cite{OKane2008}. For tranches, the given market data
consists of running spreads, upfront payments and base
correlations. To make all calculations involving both index and
tranche positions consistently use spread time series, the tranche
data are transformed into financially equivalent fair spreads without
upfront payment.

The present value of an index tranche with attachment point $K_1$,
detachment point $K_2$ and maturity $T$ is given by (cf.\
\citealp[Equation (20.1)]{OKane2008}):
\begin{equation*}
  \text{PV}(K_1,K_2)  = U(K_1,K_2) + S(K_1,K_2) \underbrace{\int_0^T
    Z(t) Q(t,K_1,K_2)\, \dd t}_{=\text{RPV01}} - \int_0^T Z(t) (-\dd
  Q(t,K_1,K_2)),  
\end{equation*}
where $U(K_1,K_2)$ is the upfront spread that is paid at inception of
the trade, $S(K_1, K_2)$ is the running spread that is paid regularly
(continuously in the case considered), $Z(t)=\e^{-r t}$ is the
time-$t$ discount factor and where $Q(t,K_1,K_2)$ denotes the {\em
  tranche survival probability}.  The expression
$\int_0^T Z(t) Q(t,K_1,K_2)\, \dd t$ denotes the {\em risky present
  value of a basis point (RPV01)} and
$\int_0^T Z(t)(-\dd Q(t,K_1,K_2))$ corresponds to the value of the
protection leg that pays in case of default. The financially
equivalent spread without an upfront payment satisfies
$s=U(K_1,K_2)/\text{RPV01} + S(K_1,K_2)$, so the calculation reduces
to calculating the {\em tranche survival curve} $Q(t,K_1,K_2)$,
$t\in [0,T]$.

Assuming $m$ contiguous tranches, the survival probabilities
$Q(T,0,K_1), Q(T,K_1,K_2), \ldots,$ $Q(T,K_{m-1},K_m)$ can be
calculated iteratively from the expected tranche losses for each
tranche (cf.\ \citealp[pp.\ 378]{OKane2008}), which in turn are
calculated from the base correlations and the index survival
probability in a one-factor Gaussian latent variable model (cf.\
\citealp[pp.\ 305--307]{OKane2008}).

Finally, the survival probabilities need to be calculated for
$t\in [0,T]$. Market convention holds that the quoted index spread
corresponds to a flat term structure \cite[p.\ 190]{OKane2008}. Making
the same assumption for the tranche spread $s$, the {\em hazard rate}
$\lambda$ entering tranche survival probabilities
$Q_T = \e^{-\lambda T}$ is determined via the so-called {\em credit
  triangle} $\lambda=s/(1-R)$, giving tranche survival probabilities
$Q_t=\e^{-\lambda t}$. The RPV01 is then calculated as
\begin{equation*}
  \int_0^T Z(t) Q(t,K_1,K_2)\, \dd t= \int_0^T \e^{-(r+\lambda)}\, \dd
  t = \frac{1-\e^{-(r+\lambda)T}}{r+\lambda}. 
\end{equation*}

\ifjbf
	\section*{References}
	\bibliographystyle{elsarticle-harv}
	\bibliography{finance,MscBIB}
\else
	\bibliographystyle{elsarticle-harv}
	\bibliography{finance,MscBIB}

\begin{thebibliography}{39}
\expandafter\ifx\csname natexlab\endcsname\relax\def\natexlab#1{#1}\fi
\expandafter\ifx\csname url\endcsname\relax
  \def\url#1{\texttt{#1}}\fi
\expandafter\ifx\csname urlprefix\endcsname\relax\def\urlprefix{URL }\fi

\bibitem[{Adams et~al.(2017)Adams, F{\"u}ss, and
  Gl{\"u}ck}]{adams2017correlations}
Adams, Z., F{\"u}ss, R., Gl{\"u}ck, T., 2017. Are correlations constant?
  {E}mpirical and theoretical results on popular correlation models in finance.
  Journal of Banking \& Finance 84, 9--24.

\bibitem[{Alexander and Sheedy(2008{\natexlab{a}})}]{alexander2008developing}
Alexander, C., Sheedy, E., 2008{\natexlab{a}}. Developing a stress testing
  framework based on market risk models. Journal of Banking \& Finance 32~(10),
  2220--2236.

\bibitem[{Alexander and Sheedy(2008{\natexlab{b}})}]{Alexander2008}
Alexander, C., Sheedy, E., 2008{\natexlab{b}}. Developing a stress testing
  framework based on market risk models. Journal of Banking \& Finance 32~(10),
  2220--2236.

\bibitem[{Ang and Bekaert(2002)}]{Ang2002a}
Ang, A., Bekaert, G., 2002. International asset allocation with regime shifts.
  The Review of Financial Studies 15~(4), 1137--1187.

\bibitem[{BCBS(2006)}]{bcbs128}
BCBS, 2006. International convergence of capital measurement and capital
  standards. Tech. rep., Basel Committee on Banking Supervision.
\newline\urlprefix\url{https://www.bis.org/publ/bcbs128.pdf}

\bibitem[{Breuer and Csisz{\'a}r(2013)}]{breuer2013systematic}
Breuer, T., Csisz{\'a}r, I., 2013. Systematic stress tests with entropic
  plausibility constraints. Journal of Banking \& Finance 37~(5), 1552--1559.

\bibitem[{Breuer et~al.(2009)Breuer, Janda{\v{c}}ka, Rheinberger, and
  Summer}]{Breuer2009}
Breuer, T., Janda{\v{c}}ka, M., Rheinberger, K., Summer, M., 2009. {How to find
  plausible, severe, and useful stress scenarios}. International Journal of
  Central Banking 5~(3), 205--224.

\bibitem[{Brigo(2002)}]{Brigo2002}
Brigo, D., May 2002. A note on correlation and rank reduction. Working Paper.
\newline\urlprefix\url{http://www.damianobrigo.it}

\bibitem[{Brigo and Mercurio(2006)}]{Brigo2006}
Brigo, D., Mercurio, F., 2006. Interest Rate Models - Theory and Practice, with
  Smile, Inflation and Credit, 2nd Edition. Springer.

\bibitem[{Buraschi et~al.(2010)Buraschi, Porchia, and
  Trojani}]{buraschi2010correlation}
Buraschi, A., Porchia, P., Trojani, F., 2010. Correlation risk and optimal
  portfolio choice. The Journal of Finance 65~(1), 393--420.

\bibitem[{Chincarini(2007)}]{Chincarini2007}
Chincarini, L.~B., 2007. The {A}maranth debacle: Failure of risk measures or
  failure of risk management? The Journal of Alternative Investments 10~(3),
  91--104.

\bibitem[{Cont and Wagalath(2016)}]{Cont2015}
Cont, R., Wagalath, L., 2016. Risk management for whales. Risk~(June), 79--82.

\bibitem[{DTCC(2011)}]{DTCC2011}
DTCC, 2011. Explanation of {Trade Information Warehouse} data.
\newline\urlprefix\url{http://dtcc.com/~/media/Files/Downloads/Settlement-Asset-services/DerivSERV/tiw_data_explanation.pdf}

\bibitem[{DTCC(2014)}]{DTCC2014}
DTCC, 2014. {Depository Trust \& Clearing Corporation -- Trade Information
  Warehouse}.
\newline\urlprefix\url{http://www.dtcc.com/repository-otc-data.aspx}

\bibitem[{Flood and Korenko(2015)}]{flood2015systematic}
Flood, M.~D., Korenko, G.~G., 2015. Systematic scenario selection: stress
  testing and the nature of uncertainty. Quantitative Finance 15~(1), 43--59.

\bibitem[{Glasserman et~al.(2015)Glasserman, Kang, and
  Kang}]{glasserman2015stress}
Glasserman, P., Kang, C., Kang, W., 2015. Stress scenario selection by
  empirical likelihood. Quantitative Finance 15~(1), 25--41.

\bibitem[{Jorion(2000)}]{Jorion2000}
Jorion, P., 2000. Risk management lessons from {Long-Term Capital Management}.
  European Financial Management 6~(3), 277--300.

\bibitem[{JPMorgan(2013)}]{JPMorgan2013}
JPMorgan, 2013. {Report of JPMorgan Chase \& Co. Management Task Force
  Regarding 2012 CIO Losses}.
\newline\urlprefix\url{http://files.shareholder.com/downloads/ONE/2272984969x0x628656/4cb574a0-0bf5-4728-9582-625e4519b5ab/Task_Force_Report.pdf}

\bibitem[{Karolyi and Stulz(1996)}]{Karolyi1996}
Karolyi, G.~A., Stulz, R.~M., 1996. Why do markets move together? an
  investigation of {US}-{J}apan stock return comovements. The Journal of
  Finance 51~(3), 951--986.

\bibitem[{Kopeliovich et~al.(2015)Kopeliovich, Novosyolov, Satchkov, and
  Schachter}]{kopeliovich2015robust}
Kopeliovich, Y., Novosyolov, A., Satchkov, D., Schachter, B., 2015. Robust risk
  estimation and hedging: A reverse stress testing approach. The Journal of
  Derivatives 22~(4), 10--25.

\bibitem[{Krishnan et~al.(2009)Krishnan, Petkova, and
  Ritchken}]{krishnan2009correlation}
Krishnan, C., Petkova, R., Ritchken, P., 2009. Correlation risk. Journal of
  Empirical Finance 16~(3), 353--367.

\bibitem[{Kupiec(1998)}]{Kupiec1998}
Kupiec, P., 1998. Stress testing in a {V}alue at {R}isk framework. The Journal
  of Derivatives 6, 7--24.

\bibitem[{Longin and Solnik(2001)}]{longin2001extreme}
Longin, F., Solnik, B., 2001. Extreme correlation of international equity
  markets. The Journal of Finance 56~(2), 649--676.

\bibitem[{Loretan and English(2000)}]{Loretan2000}
Loretan, M., English, W., June 2000. Evaluating changes in correlations during
  periods of high market volatility. BIS Quarterly Review, 29--36.

\bibitem[{McNeil et~al.(2015)McNeil, Frey, and Embrechts}]{McNeil2015}
McNeil, A., Frey, R., Embrechts, P., 2015. Quantitative Risk Management, 2nd
  Edition. Princeton University Press, Princeton, NJ.

\bibitem[{Mueller et~al.(2017)Mueller, Stathopoulos, and
  Vedolin}]{mueller2017international}
Mueller, P., Stathopoulos, A., Vedolin, A., 2017. International correlation
  risk. Journal of Financial Economics 126~(2), 270--299.

\bibitem[{Ng et~al.(2014)Ng, Li, and Yu}]{ng2014black}
Ng, F., Li, W., Yu, P.~L., 2014. A {B}lack-{L}itterman approach to correlation
  stress testing. Quantitative Finance 14~(9), 1643--1649.

\bibitem[{O'Kane(2008)}]{OKane2008}
O'Kane, D., 2008. Modelling Single-name and Multi-name Credit Derivatives.
  Wiley.

\bibitem[{Pu and Zhao(2012)}]{pu2012correlation}
Pu, X., Zhao, X., 2012. Correlation in credit risk changes. Journal of Banking
  \& Finance 36~(4), 1093--1106.

\bibitem[{Qi and Sun(2010)}]{qi2010correlation}
Qi, H., Sun, D., 2010. Correlation stress testing for value-at-risk: an
  unconstrained convex optimization approach. Computational Optimization and
  Applications 45~(2), 427--462.

\bibitem[{Rebonato(2002)}]{Rebonato2002}
Rebonato, R., 2002. Modern Pricing of Interest-Rate Derivatives: The LIBOR
  Market Model and Beyond. Princeton University Press.

\bibitem[{Rebonato(2004)}]{Rebonato2004}
Rebonato, R., 2004. Volatility and Correlation, 2nd Edition. John Wiley \&
  Sons.

\bibitem[{Schoenmakers and Coffey(2003)}]{Schoenmakers2003}
Schoenmakers, J., Coffey, B., 2003. Systematic generation of parametric
  correlation structures for the libor market model. International Journal of
  Theoretical and Applied Finance 6~(5), 507--519.

\bibitem[{Shiryaev(1996)}]{Shiryaev1996}
Shiryaev, A.~N., 1996. Probability, 2nd Edition. Springer, Berlin.

\bibitem[{Studer(1999)}]{Studer1999}
Studer, G., 1999. Risk measurement with maximum loss. Mathematical Methods of
  Operations Research 50~(1), 121--134.

\bibitem[{United-States-Senate(2013{\natexlab{a}})}]{USS2013Exhibits}
United-States-Senate, 2013{\natexlab{a}}. {JPMorgan Chase Whale Trades}: A case
  history of derivatives risks and abuses. exhibits.
\newline\urlprefix\url{https://www.hsgac.senate.gov/imo/media/doc/EXHIBITS%20(JPMC%20HRG%20-%20March%2015%202013)2.pdf}

\bibitem[{United-States-Senate(2013{\natexlab{b}})}]{USS2013Report}
United-States-Senate, 2013{\natexlab{b}}. {JPMorgan Chase Whale Trades}: A case
  history of derivatives risks and abuses. report.
\newline\urlprefix\url{http://www.hsgac.senate.gov/download/report-jpmorgan-chase-whale-trades-a-case-history-of-derivatives-risks-and-abuses-march-15-2013}

\bibitem[{Wied et~al.(2012)Wied, Kr{\"a}mer, and Dehling}]{wied2012testing}
Wied, D., Kr{\"a}mer, W., Dehling, H., 2012. Testing for a change in
  correlation at an unknown point in time using an extended functional delta
  method. Econometric Theory 28~(3), 570--589.

\bibitem[{Wooldridge(2016)}]{Wooldridge2016central}
Wooldridge, P., 2016. Central clearing predominates in {OTC} interest rate
  derivatives markets. BIS Quarterly Review, December 2016, 22--24.

\end{thebibliography}
\fi

\end{document}